\documentclass[lettersize,journal]{IEEEtran}
\usepackage{amssymb}
\usepackage{amsthm}
\usepackage{amsmath}
\usepackage{amsfonts}
\usepackage{algorithmic}
\usepackage{algorithm}
\usepackage{array}
\usepackage{booktabs}
\usepackage{cite}
\usepackage{color}
\usepackage{cleveref}
\usepackage{diagbox}
\usepackage{epsfig}
\usepackage{epstopdf}
\usepackage{graphicx}
\usepackage{multirow}
\usepackage{multicol}
\usepackage{makecell}
\usepackage{longtable}
\usepackage{subfigure}
\usepackage{textcomp}
\usepackage{stfloats}
\usepackage{url}
\usepackage{verbatim}
\usepackage{pifont}

\usepackage[T1]{fontenc}

\allowdisplaybreaks[2] 

\renewcommand\arraystretch{1.5} 
\newcommand{\PreserveBackslash}[1]{\let\temp=\\#1\let\\=\temp}
\newcolumntype{C}[1]{>{\PreserveBackslash\centering}p{#1}}
\newcolumntype{R}[1]{>{\PreserveBackslash\raggedleft}p{#1}}
\newcolumntype{L}[1]{>{\PreserveBackslash\raggedright}p{#1}}

\newtheorem{lemma}{Lemma}[section]

\newtheorem{theorem}{Theorem}[section]
\newtheorem{remark}{Remark}[section]

\newtheorem{assumption}{Assumption}[section]

\newcommand{\cmark}{\ding{51}}  
\newcommand{\xmark}{\ding{55}}  

\begin{document}

\title{FedFG: Privacy-Preserving and Robust Federated Learning via Flow-Matching Generation}

\author{~Ruiyang~Wang,~Rong~Pan,~and~Zhengan~Yao
\thanks{Ruiyang~Wang is with the School of Mathematics, Sun Yat-sen University, Guangzhou 510275, China. E-mail: wangry25@mail2.sysu.edu.cn.}
\thanks{Rong~Pan is with the School of Computer Science and Engineering, Sun Yat-sen University, Guangzhou 510006, China. E-mail: panr@sysu.edu.cn. (\textit{Corresponding author: Rong~Pan})}
\thanks{Zhengan~Yao is with the School of Mathematics, Sun Yat-sen University, Guangzhou 510275, China, and also with the Institute of Advanced Studies Hong Kong, Sun Yat-sen University, Hong Kong. E-mail: mcsyao@mail.sysu.edu.cn.}
\thanks{The source code is available at https://github.com/rywangcn/FedFG}
}

\markboth{IEEE TRANSACTIONS ON DEPENDABLE AND SECURE COMPUTING}
{Shell \MakeLowercase{\textit{et al.}}: A Sample Article Using IEEEtran.cls for IEEE Journals}


\maketitle

\begin{abstract}
Federated learning (FL) enables distributed clients to collaboratively train a global model using local private data. Nevertheless, recent studies show that conventional FL algorithms still exhibit deficiencies in privacy protection, and the server lacks a reliable and stable aggregation rule for updating the global model. This situation creates opportunities for adversaries: on the one hand, they may eavesdrop on uploaded gradients or model parameters, potentially leaking benign clients' private data; on the other hand, they may compromise clients to launch poisoning attacks that corrupt the global model. To balance accuracy and security, we propose FedFG, a robust FL framework based on flow-matching generation that simultaneously preserves client privacy and resists sophisticated poisoning attacks. On the client side, each local network is decoupled into a private feature extractor and a public classifier. Each client is further equipped with a flow-matching generator that replaces the extractor when interacting with the server, thereby protecting private features while learning an approximation of the underlying data distribution. Complementing the client-side design, the server employs a client-update verification scheme and a novel robust aggregation mechanism driven by synthetic samples produced by the flow-matching generator. Experiments on MNIST, FMNIST, and CIFAR-10 demonstrate that, compared with prior work, our approach adapts to multiple attack strategies and achieves higher accuracy while maintaining strong privacy protection.
\end{abstract}


\section{INTRODUCTION}\label{sec1}

\IEEEPARstart{F}{ederated} learning (FL) \cite{bonawitz2017practical, mcmahan2017communication} has emerged as a key paradigm for collaborative model training across distributed clients, where raw data remain local and only model updates are exchanged. This design aligns with data protection regulations such as GDPR \cite{regulation2016regulation} and HIPAA \cite{act1996health}, and has enabled deployments in security-critical domains, including healthcare, autonomous driving, and financial services \cite{rieke2020future, nguyen2021federated, dash2022federated, liu2024epffl, fan2025robust}. However, as FL is increasingly used to support high-stakes decision making, its security assumptions are being challenged by adversaries operating over open networks and in partially compromised client populations.

A central challenge is that FL must simultaneously ensure data confidentiality and model integrity.
On the one hand, adversaries can reconstruct high-fidelity private images or text from observed client gradients via privacy attacks, leading to sensitive information leakage \cite{zhu2019deep, zhao2020idlg, geiping2020inverting}; on the other hand, adversaries can control malicious clients to launch advanced poisoning attacks \cite{karimireddy2021learning, xie2020fall, fang2020local, cao2022mpaf, xie2025model}, thereby degrading global model performance and undermining system security. These two risks are tightly coupled: mechanisms that expose more information for server-side inspection may improve robustness but worsen privacy, whereas mechanisms that conceal updates can strengthen privacy but reduce the server's ability to detect malicious contributions.

Existing defenses often address only one side of this coupled problem or rely on assumptions that break down in practical non-IID settings. Cryptographic protocols can prevent plaintext exposure \cite{hu2024maskcrypt, li2025fedphe, dong2023privacy}, but their overhead can be prohibitive; more fundamentally, encrypted updates limit the server's ability to verify client behavior \cite{so2020byzantine}. Differential privacy (DP) provides formal protection by perturbing updates \cite{wei2020federated, xue2023differentially, wei2023personalized}, yet stringent privacy budgets can substantially degrade accuracy. Byzantine-robust aggregators such as Krum \cite{blanchard2017machine} and the coordinate-wise median \cite{yin2018byzantine} rely on geometric clustering of benign updates, an assumption often violated under heterogeneous client distributions, and can be bypassed by sophisticated attacks \cite{xie2020fall}. Validation-based schemes can improve robustness by screening updates \cite{cao2020fltrust, zhao2022detecting}, but typically require auxiliary clean data or trusted validation signals, which may conflict with the ``no extra data'' premise and introduce additional privacy concerns.

Motivated by this tension, we explore a different design axis: using synthetic data as a privacy--robustness bridge. Unlike raw updates, synthetic samples can be kept in plaintext to enable server-side verification, while avoiding a one-to-one correspondence with original instances, thereby reducing direct leakage risks. This perspective suggests a unified architecture in which clients provide behavioral evidence via generated samples rather than revealing sensitive feature representations or full parameter updates. As a result, the server can assess update quality and detect poisoning attempts without directly inspecting private client updates.

Building on this insight, we propose FedFG, a robust federated learning framework based on flow-matching generation that jointly preserves privacy and resists sophisticated poisoning attacks. On the client side, each model is decoupled into a private extractor and a public classifier, and a flow-matching generator is trained to replace the extractor when interacting with the server, protecting private features while learning an approximation to the local data distribution. On the server side, we design a client update verification scheme and a robust aggregation mechanism driven by synthetic samples produced by these generators, enabling anomaly detection and accuracy-aware reweighting under diverse attack strategies. The main contributions of this paper are summarized as follows:

\begin{itemize}
\item We propose FedFG, a federated learning framework that leverages flow-matching generators to unify client-side privacy protection and server-side robustness within a single design.
\item We design a server-side verification and aggregation scheme that uses synthetic samples to compute outlier and accuracy scores, detect malicious clients, and perform accuracy-aware reweighting over the remaining clients.
\item We establish a first-order stationarity guarantee for FedFG under nonconvex objectives, providing a rigorous theoretical foundation for federated learning framework that simultaneously address privacy preservation and model poisoning defense. 
\item To the best of our knowledge, this work is the first to adopt flow-matching as the generative backbone for a unified federated learning architecture capable of mitigating both privacy leakage and poisoning attacks. Extensive experiments and thorough evaluations demonstrate the effectiveness and superiority of our approach. 
\end{itemize}

The remainder of this paper is organized as follows. In~\cref{sec2}, we review related work and introduce the flow-matching model. In~\cref{sec3}, we present the proposed framework, which uses flow-matching generators as a cornerstone for client-side privacy protection and server-side robustness, and analyze its convergence in \cref{convergence}. In~\cref{sec4}, we detail the experimental settings and report results from comparative evaluations. Finally, in~\cref{sec5}, we conclude and discuss potential directions for future work.

\section{RELATED WORK}\label{sec2}
In this section, we review prior work on privacy-preserving and robust federated learning, covering major security threats and corresponding defense strategies. We also discuss generative federated learning and flow-matching models, the latter of which provides the foundation for our method.

\subsection{Privacy Leakage Protection}
Despite keeping raw data local, FL remains vulnerable to gradient inversion attacks \cite{zhu2019deep,zhao2020idlg,geiping2020inverting} that reconstruct training samples from shared updates. Consequently, explicit protection mechanisms are required, typically involving trade-offs among accuracy, robustness, and communication overhead. To mitigate privacy leakage caused by directly uploading local updates, existing defenses are commonly categorized into three groups: cryptography-based methods, differential privacy-based methods, and model-masking methods.

\subsubsection{Cryptography-based protection}
Cryptography-based methods enable aggregation without exposing plaintext updates, blocking gradient inversion at the protocol level. Recent studies have explored improvements along the efficiency--security trade-off. Hu and Li \cite{hu2024maskcrypt} propose a gradient-guided mask-selection mechanism that encrypts only a small portion of updates while maintaining protection. FedPHE \cite{li2025fedphe} builds contribution-aware secure aggregation based on homomorphic encryption with sparsification and obfuscation techniques. Dong et al. \cite{dong2023privacy} leverage secure multi-party computation to implement Byzantine-robust aggregation while preserving parameter privacy. Nevertheless, cryptography-based protection is not foolproof. Lam et al. \cite{lam2021gradient} show that even with secure aggregation, the server may still infer individual updates under certain system conditions.

\subsubsection{Differential privacy-based protection}
Differential privacy (DP) provides a quantifiable privacy guarantee by injecting noise into client updates. A key advantage of DP is computational efficiency; however, tighter privacy budgets typically reduce model utility. Xue et al. \cite{xue2023differentially} propose an adaptive noise mechanism that allocates noise according to parameter sensitivity, achieving higher accuracy under the same privacy budget. Wei et al. \cite{wei2023personalized} develop a privacy-budget allocation scheme for personalized FL based on composition theory and derive convergence bounds. Yang et al. \cite{yang2023dynamic} propose dynamic personalization with adaptive constraints to reduce the adverse impact of noise on critical parameters.

\subsubsection{Model masking methods}
Model masking methods avoid uploading complete updates by concealing a subset of network parameters. Chen et al. \cite{chen2023fedrsm} propose FedRSM, which applies representational similarity analysis to estimate layer importance and selectively hide a subset of layers. Building on this, DPR-PPFL \cite{chen2024data} extends the approach to IoT settings using layer-wise and model-based similarity analysis to resist both inversion and poisoning attacks. Another strategy is to split the network and conceal feature information. Wu et al. \cite{ijcai2022-324} adopt split learning by partitioning the local model into a private extractor and a public classifier, and employ a cGAN generator as a proxy for the extractor during server interaction, thereby mitigating privacy attacks while maintaining accuracy. Moreover, Ma et al. \cite{ma2024ppidsg} map images into a target domain via block encryption and train a GAN to model the distribution in that domain, uploading only the generator parameters to support collaboration. Overall, masking-based methods are typically more computation-friendly and can enable the server to perform limited quality assessment using available plaintext information; however, sustaining stable generalization under non-IID data distributions requires careful mechanism design.

\subsection{Model Poisoning Defense}
In model poisoning attacks, adversaries manipulate local updates to degrade global model performance. Representative strategies include sign flipping \cite{karimireddy2021learning}, inner product manipulation \cite{xie2020fall}, and optimization-based attacks \cite{fang2020local,cao2022mpaf,xie2025model} that maximize deviation or steer the model toward attacker-chosen low-accuracy states over multiple rounds. To defend against model poisoning in federated learning, existing studies can be broadly grouped into three lines: statistically robust aggregation, similarity-based detection, and validation-based filtering.

\subsubsection{Statistical robust aggregation}
Statistical approaches rely on robust statistics to suppress the impact of outlier updates. Representative examples include Krum \cite{blanchard2017machine}, which selects the update closest to the majority; coordinate-wise Median and TrimmedMean \cite{yin2018byzantine}, which compute per-coordinate statistics after removing extremes; and geometric-median-based aggregators \cite{pillutla2022robust}, which improve the breakdown point but may suffer irreducible error under highly heterogeneous data. While simple to implement, these methods can be sensitive to non-IID heterogeneity and may degrade when benign clients naturally exhibit diverse update patterns.

\subsubsection{Similarity-based detection}
Similarity-based methods exploit consistency patterns that coordinated attackers often exhibit. FoolsGold \cite{fang2020local} addresses Sybil poisoning by adaptively down-weighting clients that exhibit high gradient similarity across rounds, under the assumption of colluding attackers. DPR-PPFL \cite{chen2024data} moves similarity assessment from parameter space to representation space, applying representational similarity analysis to measure consistency and filter anomalous models in IoT scenarios.

\subsubsection{Validation-based filtering}
Validation-based methods authenticate updates using trusted or synthesized validation signals. FLTrust \cite{cao2020fltrust} assumes the server holds a small clean dataset and assigns trust scores based on directional agreement. GANDefense \cite{zhao2022detecting} generates audit data with a server-side GAN and removes malicious updates using an audit-accuracy threshold. GAN-Filter \cite{zafar2025robust} avoids external datasets by letting the global model guide synthetic sample generation for adaptive filtering. However, these approaches often require auxiliary datasets or depend on an unpoisoned global model to guide the generator, which may conflict with strict FL privacy assumptions or fail under early-round attacks.

\subsection{Generative Federated Learning and Flow-Matching Models}

Existing generative federated learning (GFL) methods often address privacy or robustness in isolation rather than jointly. Client-side generation approaches \cite{ijcai2022-324,ma2024ppidsg} use generators to replace privacy-sensitive representations, but typically lack robust server-side verification mechanisms. Server-side generation approaches \cite{zhao2022detecting,zafar2025robust} produce audit samples for scoring client updates, but may depend on auxiliary clean validation data and can be unstable under non-IID heterogeneity. These limitations motivate a unified architecture in which the generative component supports both privacy protection and robust aggregation. 

Recent flow-matching methods provide a promising foundation by enabling simulation-free training of continuous normalizing flows (CNFs) \cite{lipman2023flow,tong2024improving}. A CNF defines a deterministic probability-flow ODE
\begin{equation}
\frac{d x_{t}}{dt}=v_\theta(t,x_t),
\end{equation}
which transports samples from a source distribution to the data distribution. Flow matching trains $v_{\theta}$ via regression to a target vector field:
\begin{equation}
\mathcal{L}_{\mathrm{FM}}(\theta)=\mathbb{E}_{t,x\sim p_t}\big\|v_\theta(t,x)-u_t(x)\big\|_2^2 .
\end{equation}
Importantly, optimal-transport displacement interpolation can yield straighter conditional trajectories, enabling faster and more stable generation with fewer neural function evaluations \cite{lipman2023flow,tong2024improving,liu2022flow}. These properties make flow matching well suited to FL: it supports stable sample synthesis under non-IID settings and produces plaintext synthetic data that the server can use for update verification and robust aggregation. Accordingly, we adopt flow matching as the generative backbone and propose FedFG, which tightly couples client-side privacy protection with server-side robustness via flow-matching generated synthetic samples.

\begin{table}[htbp]
\centering
\caption{Key Notations for the Proposed Method}
\label{tab:notations}
\renewcommand{\arraystretch}{1.3} 
\begin{tabular}{cp{0.75\columnwidth}}
\hline\hline
Symbol & Meaning \\
\hline
$N$ & Number of clients \\
\hline
$K$ & Number of classes \\
\hline
$R$ & Number of communication rounds \\
\hline
$\mathcal{B}, \mathcal{M}$ & Set of benign clients and set of malicious clients \\
\hline
$\mathcal{D}_i$ & Local dataset of client $i$ \\
\hline
$E_i$ & Private feature extractor of client $i$ \\
\hline
$C_i, FG_i$ & Public classifier and flow-matching generator of client $i$ \\
\hline
$C_g, FG_g$ & Global classifier and global flow-matching generator \\
\hline
$\theta_{(\cdot)}$ & Model parameters of the corresponding network module \\
\hline
$z \sim p(z)$ & Noise sampled from the prior distribution \\
\hline
$t \in [0,1]$ & Time variable in flow-matching \\
\hline
$h, \tilde{h}$ & Real feature representation and synthetic feature probe \\ 
\hline
$h_0, h_1, h_t$ & Initial noise feature, target real feature, and intermediate feature on the flow path \\
\hline
$v_{\theta_{FG_i}}(\cdot)$ & Conditional vector field of the generator \\
\hline
$u_t(\cdot)$ & Conditional target flow for flow-matching \\
\hline
$\mathcal{L}_{\mathrm{cls}}, \mathcal{L}_{\mathrm{FM}}$ & Classification loss and flow-matching loss \\
\hline
$w_i^{(r)}$ & Aggregation weight of client $i$ at round $r$. \\
\hline
$s_i^{r}$ & Accuracy score of client $i$ at round $r$ \\ 
\hline
$\alpha_i^{r}$ & Relative accuracy score of client $i$ at round $r$ \\
\hline
$\bar{\alpha}_i^r$ & Renormalized accuracy weight for benign client $i$ \\
\hline
$p_i$ & Predictive probability vector of client $i$ \\
\hline
$d_{\mathrm{H}}$ & Hellinger distance between client classifiers \\
\hline
$o_i^{r}$ & Outlier score of client $i$ at round $r$ \\
\hline
$\tau^{r}$ & Adaptive threshold computed via the Hampel rule \\ 
\hline
$\gamma$ & Tunable cutoff parameter for the Hampel rule \\
\hline
$\kappa$ & Accuracy threshold for filtering \\
\hline\hline
\end{tabular}
\end{table}


\section{PROPOSED METHOD}\label{sec3}

 \begin{figure*}[!htbp]
  \centering
    \includegraphics[width=0.9\textwidth]{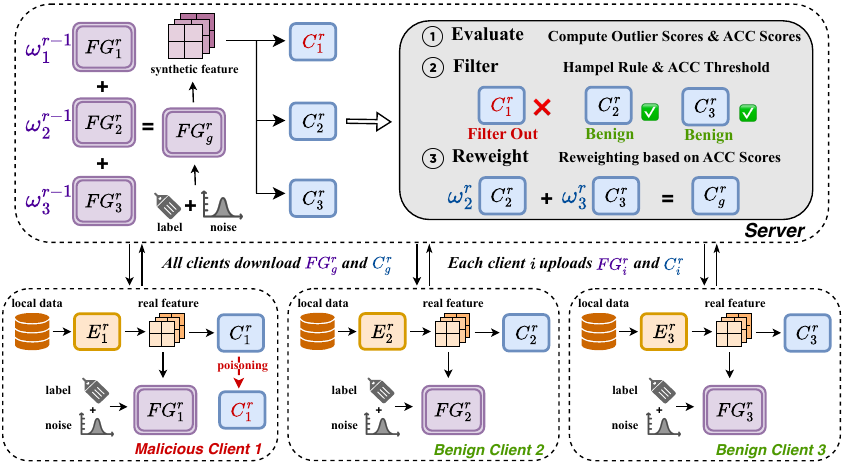}
    \caption{Overview of FedFG. On each client, the local model is decoupled into a private feature extractor and a public classifier, and is equipped with a flow-matching generator that replaces the extractor during communication with the server. The server collects model parameters from clients and performs robust aggregation with three modules: 1) evaluation of outlier and accuracy scores via synthetic samples generated by the flow-matching generator; 2) detection of malicious clients based on the Hampel rule and accuracy threshold; and 3) accuracy-aware reweighting aggregation over the remaining benign clients.}
    \label{Fig-overview}
 \end{figure*}

In this section, we present FedFG, a secure and privacy-preserving federated learning framework based on flow-matching generation. As illustrated in \figurename~\ref{Fig-overview}, the framework consists of two iterative procedures.

We consider a standard cross-silo FL system consisting of a central server and $N$ clients. Client $i$ holds a private dataset $\mathcal{D}_i=\{(x,y)\}$, where $x$ and $y\in\{1,\dots,K\}$ denote an input sample and its label, respectively.
FedFG decouples each client model into a private feature extractor $E_i: \mathcal{X}\rightarrow \mathbb{R}^d$ and a public classifier $C_i: \mathbb{R}^d\rightarrow \mathbb{R}^K$, and further equips each client with a flow-matching generator $FG_i: \mathcal{Z}\times\{1,\dots,K\}\rightarrow \mathbb{R}^d$ that approximates the feature distribution induced by the extractor, where $\mathcal{Z}$ is a simple noise distribution (e.g., a standard Gaussian). The server maintains a global generator $FG_g$ and a global classifier $C_g$. In each communication round $r$, clients download $(\theta_{FG_g}^{r-1},\theta_{C_g}^{r-1})$, perform local updates, and upload $(\theta_{FG_i}^{r},\theta_{C_i}^{r})$. The server then verifies client updates using synthetic feature probes produced by $FG_g$ and performs robust aggregation to obtain $(\theta_{FG_g}^{r},\theta_{C_g}^{r})$. Details are provided below for both the client-side and server-side procedures.

\subsection{Client Side: Privacy-Preserving Local Update}
\label{subsec:client}

A key privacy risk in FL stems from sharing parameters that directly interact with raw data, enabling gradient inversion and related reconstruction attacks. FedFG mitigates this risk by keeping the extractor $E_i$ strictly local. Only $(\theta_{FG_i},\theta_{C_i})$ are exchanged with the server. Since $FG_i$ operates in feature space and does not process raw inputs, and since $C_i$ maps features to logits, the information directly exposed about raw samples is substantially reduced compared with sharing an end-to-end model. Algorithm~\ref{alg:client_fedfg} describes the procedure in detail.

At round $r$, client $i$ initializes its classifier from the global classifier and then updates $\theta_{E_i}$ and $\theta_{C_i}$ on its private data by minimizing the supervised classification loss:
\begin{equation}
\mathcal{L}_{\mathrm{cls}}(\theta_{E_i},\theta_{C_i})
=\mathbb{E}_{(x,y)\sim \mathcal{D}_i}
\Big[\Omega\big(C_i(E_i(x;\theta_{E_i});\theta_{C_i}),\,y\big)\Big],
\label{eq:client-cls}
\end{equation}
where $\Omega(\cdot,\cdot)$ denotes the cross-entropy loss. The local update is:
\begin{equation}
(\theta_{E_i},\theta_{C_i})
\leftarrow (\theta_{E_i},\theta_{C_i})-\eta_1\nabla \mathcal{L}_{\mathrm{cls}},
\label{eq:client-ec-update}
\end{equation}
with learning rate $\eta_1$.

After updating the extractor, client $i$ trains a flow-matching generator $FG_i$ to approximate the conditional feature distribution of the extractor outputs given label $y$.
Let 
\begin{equation}\label{eq:h}
    h=E_i(x;\theta_{E_i})
\end{equation}
denote the real feature. Client $i$ optimizes $\theta_{FG_i}$ via a flow-matching objective:
\begin{equation}
\!\mathcal{L}_{\mathrm{FM}}\!\big(\!\theta_{FG_i};\!\theta_{E_i}\!\big)
\!=\!\mathbb{E}_{(x,y)\sim \mathcal{D}_i,z\sim p(z)}\!
\Big[\ell_{\mathrm{FM}}\!\big(\!FG_i,\!E_i(x),\!y,\!z\!\big)\!\Big],
\label{eq:client-fm-loss}
\end{equation}
where $\ell_{\mathrm{FM}}(\cdot)$ is the per-sample flow-matching loss implemented by the client-side generator module.
The generator update is:
\begin{equation}
\theta_{FG_i}\leftarrow \theta_{FG_i}-\eta_2\nabla \mathcal{L}_{\mathrm{FM}},
\label{eq:client-g-update}
\end{equation}
with learning rate $\eta_2$.

To instantiate a standard flow-matching model, we parameterize the generator as a conditional, time-dependent vector field $v_{\theta_{FG_i}}(h,t,y)$ and define generation through the ordinary differential equation
\begin{equation}
\frac{d h(t)}{dt} = v_{\theta_{FG_i}}(h(t),t,y),\quad t\in[0,1],
\label{eq:fm-ode}
\end{equation}
where the initial state $h(0)$ is sampled from a simple prior induced by $z\sim p(z)$ and a lightweight mapping from noise space to feature space. Sampling corresponds to integrating \eqref{eq:fm-ode} from $t=0$ to $t=1$ to obtain $h(1)$.

For training, we adopt an independent conditional flow-matching probability path between a noise sample $h_0$ and a real feature $h_1 =h$ as mentioned in Eq. (\ref{eq:h}):
\begin{align}
&h_t = t h_1 + (1-t)h_0 + \sigma \epsilon,  \label{eq:fm-path1} \\
&h_0\sim \mathcal{N}(0,I),\ \epsilon\sim \mathcal{N}(0,I),\ t\sim \mathrm{Unif}(0,1), \label{eq:fm-path2}
\end{align}
where $\sigma\ge 0$ is a scalar that controls the stochasticity of the path.
Under this construction, the conditional target flow is
\begin{equation}
u_t(h_1\mid h_0)=h_1-h_0,
\label{eq:fm-target-flow}
\end{equation}
and the per-sample flow-matching loss can be written as the mean-squared error between the predicted and target flows:
\begin{equation}
\ell_{\mathrm{FM}}\!\big(FG_i,\!h_1,\!y,\!z\big)\!=\!\mathbb{E}_{t}\!\Big[\big\|v_{\theta_{FG_i}}(h_t,\!t,\!y)-u_t(h_1\!\mid\! h_0)\big\|_2^2\Big].
\label{eq:fm-per-sample}
\end{equation}
The condition $y$ is injected into $v_{\theta_{FG_i}}$ via a learnable label embedding, and time $t$ is provided through a suitable time embedding, enabling the model to learn class-specific transport dynamics in the extractor feature space. During this stage, $\theta_{E_i}$ is fixed and only $\theta_{FG_i}$ is updated, so the generator learns the conditional feature distribution induced by the current extractor.

After the two-stage update, client $i$ uploads only the public modules $(\theta_{FG_i}^{r},\,\theta_{C_i}^{r})$, while keeping $\theta_{E_i}^{r}$ local. This design preserves privacy by preventing the server from accessing the module that directly transforms raw inputs into intermediate representations.

\begin{algorithm}[t]
\caption{Client-Side Procedure of FedFG}
\label{alg:client_fedfg}
\begin{algorithmic}[1]
\REQUIRE Local dataset $\mathcal{D}_i$; communication rounds $R$; local epochs $E_{\mathrm{loc}}$; flow epochs $E_{\mathrm{flow}}$; learning rates $\eta_1,\eta_2$; malicious client set $\mathcal{M}$.

\FOR{$r = 1,2,\ldots,R$}

    \STATE /* Step 1: Extractor and classifier training */
    \FOR{$e = 1,2,\ldots,E_{\mathrm{loc}}$}
        \FOR{each minibatch $(x,y)\sim \mathcal{D}_i$}
            \STATE Compute $\mathcal{L}_{\mathrm{cls}}$ using Eq.~\eqref{eq:client-cls}
            \STATE Update $(\theta_{E_i},\theta_{C_i})$ using Eq.~\eqref{eq:client-ec-update}
        \ENDFOR
    \ENDFOR

    \STATE /* Step 2: Flow-matching generator training */
    \FOR{$e = 1,2,\ldots,E_{\mathrm{flow}}$}
        \FOR{each minibatch $(x,y)\sim \mathcal{D}_i$}
            \STATE Extract real features $h_1 \leftarrow E_i(x;\theta_{E_i})$
            \STATE Sample $t, h_0, \epsilon$ using Eq.~\eqref{eq:fm-path2} and form $h_t$ using Eq.~\eqref{eq:fm-path1}
            \STATE Compute the FM loss $\ell_{\mathrm{FM}}$ using Eq.~\eqref{eq:fm-per-sample}
            \STATE Update $\theta_{FG_i}$ using Eq.~\eqref{eq:client-g-update}
        \ENDFOR
    \ENDFOR

    \STATE /* Step 3: Model poisoning by malicious clients */
    \IF{$i \in \mathcal{M}$}
        \STATE $(\theta_{FG_i}^{r},\,\theta_{C_i}^{r}) \leftarrow \mathcal{A}\big(\theta_{FG_i}^{r},\,\theta_{C_i}^{r};\,\theta_{FG_g}^{r-1},\,\theta_{C_g}^{r-1}\big)$
    \ENDIF

    \STATE /* Step 4: Upload public components */
    \STATE Send $(\theta_{FG_i}^{r},\,\theta_{C_i}^{r})$ to the server

    \STATE /* Step 5: Download global models */
    \STATE Receive $(\theta_{FG_g}^{r},\,\theta_{C_g}^{r})$ from the server
    \STATE Update public modules: $\theta_{C_i}^{r} \leftarrow \theta_{{C_g}}^{r}$, \ $\theta_{FG_i}^{r} \leftarrow \theta_{FG_g}^{r}$
\ENDFOR

\STATE \textbf{return} $\theta_{E_i}^{R},\, \theta_{C_i}^{R},\, \theta_{FG_i}^{R}$
\end{algorithmic}
\end{algorithm}

\subsection{Server Side: Update Verification and Robust Aggregation}
\label{subsec:server}
As illustrated in \figurename~\ref{Fig-overview}, malicious clients ($i\in\mathcal{M}$) may apply a poisoning function $\mathcal{A}(\cdot)$ to their public parameters before uploading (e.g., inner product manipulation), aiming to degrade the global model accuracy through corrupted updates. To counter such threats, FedFG verifies client updates using synthetic feature probes generated by the aggregated generator and filters outliers prior to aggregation. The key idea is to let $FG_g$ act as a judge that produces class-conditional synthetic features, which are then fed into each client classifier to obtain comparable predictive distributions for distinguishing benign updates from poisoned ones.

The server first forms preliminary global models by weighted averaging:
\begin{equation}
\theta_{FG_g}^{r} \leftarrow \sum_{i=1}^{N} w_i^{(r-1)}\,\theta_{FG_i}^{r},\quad
\theta_{C_g}^{r} \leftarrow \sum_{i=1}^{N} w_i^{(r-1)}\,\theta_{C_i}^{r},
\label{eq:server-preagg}
\end{equation}
where $w_i^{(r-1)}$ is the server's current aggregation weight for client $i$ (initialized proportionally to local data sizes and later updated using accuracy scores; see below). This preliminary $FG_g^{r}$ is then used to generate synthetic feature probes for verification.

The server samples labels $y\sim \mathrm{Unif}(\{1,\dots,K\})$ and noise $z\sim p(z)$, and generates synthetic features
\begin{equation}
\tilde{h}=FG_g^{r}(z,y;\theta_{FG_g}^{r}).
\label{eq:synth-feature}
\end{equation}
For each client classifier $C_i^{r}$, the server evaluates its correctness on these synthetic pairs:
\begin{equation}
s_i^{r}
=
\mathbb{E}_{y,z}
\Big[\mathbb{I}\big(\arg\max C_i^{r}(\tilde{h}) = y\big)\Big],
\label{eq:acc-score}
\end{equation}
where $\mathbb{I}(\cdot)$ is the indicator function and the expectation is approximated via Monte Carlo batches.
The scores are normalized to obtain the relative accuracy score:
\begin{equation}
\alpha_i^{r}=\frac{s_i^{r}}{\sum_{j=1}^{N} s_j^{r}+\varepsilon}.
\label{eq:relacc-score}
\end{equation}

Using the same synthetic probes $(\tilde{h},y)$, the server computes the predictive probability vector of client $i$:
\begin{equation}
p_i(\tilde{h})=\mathrm{softmax}\big(C_i^{r}(\tilde{h})\big)\in\Delta^{K-1}.
\label{eq:prob}
\end{equation}
For two clients $i$ and $j$, their predictive discrepancy is measured by the Hellinger distance:
\begin{equation}
d_{\mathrm{H}}\big(p_i,p_j\big)
=
\frac{1}{\sqrt{2}}
\left\|
\sqrt{p_i}-\sqrt{p_j}
\right\|_2.
\label{eq:hellinger}
\end{equation}
We define the outlier score of client $i$ as the average pairwise distance to all other clients:
\begin{equation}
o_i^{r}=\frac{1}{N-1}\sum_{j\neq i}\mathbb{E}_{y,z}\Big[d_{\mathrm{H}}\big(p_i(\tilde{h}),p_j(\tilde{h})\big)\Big].
\label{eq:outlier}
\end{equation}
Intuitively, poisoning often yields classifiers whose predictive behavior on shared probes deviates from that of benign clients, leading to larger $o_i^{r}$.

To robustly identify outliers without assuming a parametric distribution, the server applies a Hampel rule based on the median and MAD:
\begin{align}
m^{r} &= \mathrm{median}\{o_i^{r}\}_{i=1}^{N}, \\
\mathrm{MAD}^{r} &= \mathrm{median}\{|o_i^{r}-m^{r}|\}_{i=1}^{N}+\varepsilon, \\
\tau^{r} &= m^{r}+\gamma \cdot 1.4826\cdot \mathrm{MAD}^{r}.
\label{eq:hampel}
\end{align}
Here, $\gamma$ is a tunable cutoff in units of the MAD-based scale, and $1.4826$ rescales $\mathrm{MAD}$ to be $\sigma$-equivalent under Gaussian noise. Client $i$ is flagged as an outlier if $o_i^{r}>\tau^{r}$.

In addition to the outlier test, FedFG filters low-accuracy updates by thresholding the relative accuracy scores: client $i$ is filtered if $\alpha_i^{r}<\kappa$. The two filters are combined by set union.

Let $\mathcal{B}^{r}$ denote the set of remaining benign clients after filtering. The server re-normalizes the accuracy weights over $\mathcal{B}^{r}$:
\begin{equation}
\bar{\alpha}_i^{r}=\frac{\alpha_i^{r}}{\sum_{j\in \mathcal{B}^{r}}\alpha_j^{r}+\varepsilon},\qquad i\in \mathcal{B}^{r}.
\label{eq:renorm}
\end{equation}
Finally, the server aggregates only benign updates:
\begin{equation}
\theta_{FG_g}^{r} \leftarrow \sum_{i\in \mathcal{B}^{r}} \bar{\alpha}_i^{r}\,\theta_{FG_i}^{r},\quad
\theta_{C_g}^{r} \leftarrow \sum_{i\in \mathcal{B}^{r}} \bar{\alpha}_i^{r}\,\theta_{C_i}^{r}.
\label{eq:robust-agg}
\end{equation}
The resulting $(\theta_{FG_g}^{r},\theta_{C_g}^{r})$ are broadcast to all clients for the next round. Moreover, we set the next-round prior weights as $w_i^{(r)}=\alpha_i^{r}$, yielding an adaptive, history-aware aggregation scheme. The full procedure is presented in Algorithm~\ref{alg:server_fedfg}.

\begin{algorithm}[t]
\caption{Server-Side Procedure of FedFG}
\label{alg:server_fedfg}
\begin{algorithmic}[1]
\REQUIRE Number of clients $N$; communication rounds $R$; accuracy threshold $\kappa$; Hampel cutoff $\gamma$.

\FOR{$r = 1,2,\ldots,R$}
    \STATE /* Step 1: Collect client updates */
    \STATE Receive $\{(\theta_{C_i}^{r},\theta_{FG_i}^{r})\}_{i=1}^{N}$

    \STATE /* Step 2: Preliminary aggregation */
    \STATE Compute preliminary $(\theta_{FG_g}^{r},\theta_{C_g}^{r})$ using Eq.~\eqref{eq:server-preagg}

    \STATE /* Step 3: Synthetic feature probes */
    \STATE Sample $(z,y)$ and generate synthetic features $\tilde{h}$ using Eq.~\eqref{eq:synth-feature}

    \STATE /* Step 4: Evaluate client updates on probes */
    \FOR{$i = 1$ to $N$}
        \STATE Compute accuracy score $s_i^{r}$ using Eq.~\eqref{eq:acc-score}
        \STATE Compute relative accuracy score $\alpha_i^{r}$ using Eq.~\eqref{eq:relacc-score}
        \STATE Compute outlier score $o_i^{r}$ using Eqs.~\eqref{eq:hellinger} and~\eqref{eq:outlier}
    \ENDFOR

    \STATE /* Step 5: Malicious client detection */
    \STATE Compute threshold $\tau^{r}$ using Eq.~\eqref{eq:hampel}
    \STATE $\mathcal{B}^{r} \leftarrow \{\, i \mid (o_i^{r} < \tau^{r})\ \wedge\ (\alpha_i^{r} > \kappa)\,\}$

    \STATE /* Step 6: Accuracy-aware reweighting */
    \STATE For $i\in\mathcal{B}^{r}$, compute renormalized weights $\bar{\alpha}_i^{r}$ using Eq.~\eqref{eq:renorm}
    \STATE Aggregate $(\theta_{FG_g}^{r},\theta_{C_g}^{r})$ over $\mathcal{B}^{r}$ using Eq.~\eqref{eq:robust-agg}
    \STATE Update next-round prior weights: $w_i^{(r)} \leftarrow \alpha_i^{r}$

    \STATE /* Step 7: Broadcast global models */
    \STATE Send $(\theta_{FG_g}^{r},\,\theta_{C_g}^{r})$ to all clients
\ENDFOR

\STATE \textbf{return} $\theta_{C_g}^{R},\, \theta_{FG_g}^{R}$
\end{algorithmic}
\end{algorithm}


\section{Convergence Analysis}\label{convergence}

We establish a first-order stationarity guarantee for FedFG. Let $\mathcal{B}$ and $\mathcal{M}$ partition the $N$ clients into benign and malicious sets. We stack the server-maintained public parameters as $\theta^r := [\theta_{C_g}^r;\,\theta_{FG_g}^r] \in \mathbb{R}^p$ and define each benign client's local objective as $F_i(\theta) := \mathbb{E}_{\xi \sim \mathcal{D}_i}[\ell(\theta;\xi)]$ for $i \in \mathcal{B}$. The accuracy-weighted global objective at round $r$ is
\begin{equation}
  \widetilde{F}_r(\theta) := \sum_{i \in \mathcal{B}} \bar{\alpha}_i^r\, F_i(\theta),
  \label{eq:conv-Ft}
\end{equation}
where $\bar{\alpha}_i^r \ge 0$, $\sum_{i \in \mathcal{B}} \bar{\alpha}_i^r = 1$ are the renormalized accuracy-aware weights from Eq.~\eqref{eq:renorm}. Let $F_{*,r}(\theta) := \sum_{i \in \mathcal{B}} \bar{\alpha}_{i,*}^r\, F_i(\theta)$ denote the population-weight counterpart obtained with infinitely many synthetic probes, and let $\mathcal{E}^r := \{\mathcal{B}^r = \mathcal{B}\}$ be the event that the server correctly identifies all benign clients at round $r$.

We introduce the following assumptions required for the analysis. Assumptions~\ref{ass:smooth}--\ref{ass:local} are standard in federated optimization~\cite{Li2020On, karimireddy2020scaffold}, while Assumptions~\ref{ass:verif}--\ref{ass:weight} capture the specific properties of FedFG's verification mechanism.

\begin{assumption}[Smoothness and boundedness]
\label{ass:smooth}
For each $i \in \mathcal{B}$, $F_i$ is $L$-smooth, lower bounded by $F^{\inf} > -\infty$, and has bounded gradients: $\|\nabla F_i(\theta)\| \le G$ for all $\theta$.
\end{assumption}

\begin{assumption}[Stochastic gradients]
\label{ass:sgd}
Each benign client samples unbiased stochastic gradients with bounded variance: $\mathbb{E}[g_i(\theta;\xi)] = \nabla F_i(\theta)$ and $\mathbb{E}\|g_i(\theta;\xi) - \nabla F_i(\theta)\|^2 \le \sigma^2$.
\end{assumption}

\begin{assumption}[Bounded heterogeneity]
\label{ass:het}
Defining $F_{\mathcal{B}}(\theta) := \frac{1}{|\mathcal{B}|}\sum_{i \in \mathcal{B}} F_i(\theta)$, we have
\begin{equation}
  \frac{1}{|\mathcal{B}|}\sum_{i \in \mathcal{B}} \|\nabla F_i(\theta) - \nabla F_{\mathcal{B}}(\theta)\|^2 \le \zeta^2,
  \quad \forall\,\theta.
\end{equation}
\end{assumption}

\begin{assumption}[Local update]
\label{ass:local}
Let $Q \ge 1$ denote the number of local SGD steps per round, and let $\eta_r > 0$ be the local learning rate at round $r$. Conditioned on $(\theta^r, \mathcal{E}^r)$, the aggregated benign update $\Delta_{\mathcal{B}}^r := \theta^{r+1} - \theta^r$ satisfies, for constants $B, V \ge 0$ independent of $r$:
\begin{align}
  \big\|&\mathbb{E}[\Delta_{\mathcal{B}}^r \mid \theta^r, \mathcal{E}^r]
  + \eta_r Q\,\nabla \widetilde{F}_r(\theta^r)\big\|
  \le B\,\eta_r^2\, Q^2\,\zeta, \label{eq:conv-local-bias} \\
  &\mathbb{E}\big[\|\Delta_{\mathcal{B}}^r\|^2 \mid \theta^r, \mathcal{E}^r\big]
  \le V\,\eta_r^2\,(Q\sigma^2 + Q^2\zeta^2). \label{eq:conv-local-var}
\end{align}
\end{assumption}

\begin{assumption}[Robust verification]
\label{ass:verif}
For all $r$, $\mathbb{P}(\mathcal{E}^r) \ge 1 - \delta_f$ for some $\delta_f \in (0,1)$. On $\neg\mathcal{E}^r$, the possibly adversarial update satisfies $\mathbb{E}[\|\Delta_{\mathcal{M}}^r\|^2 \mid \neg\mathcal{E}^r] \le G_{\max}^2$.
\end{assumption}

\begin{assumption}[Weight estimation and objective stability]
\label{ass:weight}
Using $S$ i.i.d.\ synthetic probes, $\mathbb{E}[\|\bar{\alpha}^r - \bar{\alpha}_*^r\|_1^2] \le C_\alpha / S$ for some constant $C_\alpha > 0$.
Moreover, $\mathbb{E}[\sup_\theta |\widetilde{F}_{r+1}(\theta) - \widetilde{F}_r(\theta)|] \le \Delta_r$, with $\Delta_{0:R} := \sum_{r=0}^{R-1}\Delta_r < \infty$.
\end{assumption}

Based on these assumptions, we present two supporting lemmas.

\begin{lemma}[One-step descent]
\label{lem:descent}
Under Assumptions~\ref{ass:smooth} and~\ref{ass:local}, if $\eta_r Q \le \min\{1/(2L),\,1\}$, then on the event $\mathcal{E}^r$:
\begin{equation}
\begin{aligned}
  \mathbb{E}\big[\widetilde{F}_r(\theta^{r+1}) \mid \theta^r, \mathcal{E}^r\big]
  \le\;& \widetilde{F}_r(\theta^r)
  - c_1\,\eta_r Q\,\|\nabla \widetilde{F}_r(\theta^r)\|^2 \\
  &+ c_2\,\eta_r^2\,(Q\sigma^2 + Q^2\zeta^2),
\end{aligned}
\label{eq:conv-one-step}
\end{equation}
where $c_1 = 3/4$ and $c_2 = LV/2 + B^2$.
\end{lemma}

\begin{proof}
By $L$-smoothness and the conditional expectation,
\begin{equation}
\begin{aligned}
  &\mathbb{E}\big[\widetilde{F}_r(\theta^{r+1}) \mid \theta^r, \mathcal{E}^r\big] \\
  &\le \widetilde{F}_r(\theta^r)
  + \big\langle \nabla \widetilde{F}_r(\theta^r),\,
    \mathbb{E}[\Delta_{\mathcal{B}}^r \mid \theta^r, \mathcal{E}^r]\big\rangle \\
  &\quad + \frac{L}{2}\,
    \mathbb{E}\big[\|\Delta_{\mathcal{B}}^r\|^2 \mid \theta^r, \mathcal{E}^r\big].
\end{aligned}
\label{eq:conv-smooth-expand}
\end{equation}
Define $b_r := \mathbb{E}[\Delta_{\mathcal{B}}^r \mid \theta^r, \mathcal{E}^r] + \eta_r Q\,\nabla \widetilde{F}_r(\theta^r)$, so that $\|b_r\| \le B\,\eta_r^2\,Q^2\,\zeta$ by~\eqref{eq:conv-local-bias}. The inner product term becomes
\begin{equation}
\begin{aligned}
  &\big\langle \nabla \widetilde{F}_r(\theta^r),\,
    \mathbb{E}[\Delta_{\mathcal{B}}^r \mid \theta^r, \mathcal{E}^r]\big\rangle \\
  &= -\eta_r Q\,\|\nabla \widetilde{F}_r(\theta^r)\|^2
  + \big\langle \nabla \widetilde{F}_r(\theta^r),\, b_r\big\rangle.
\end{aligned}
\end{equation}
Applying Young's inequality $\langle a, b\rangle \le \frac{\alpha}{2}\|a\|^2 + \frac{1}{2\alpha}\|b\|^2$ with $\alpha = \eta_r Q / 2$ gives
\begin{equation}
  \big\langle \nabla \widetilde{F}_r(\theta^r),\, b_r\big\rangle
  \le \frac{\eta_r Q}{4}\|\nabla \widetilde{F}_r(\theta^r)\|^2
  + \frac{1}{\eta_r Q}\|b_r\|^2.
\end{equation}
Since $\eta_r Q \le 1$, we have $\frac{1}{\eta_r Q}\|b_r\|^2 \le B^2\,\eta_r^3\,Q^3\,\zeta^2 \le B^2\,\eta_r^2\,Q^2\,\zeta^2$. Substituting back and combining with the second-moment bound~\eqref{eq:conv-local-var} in~\eqref{eq:conv-smooth-expand} yields~\eqref{eq:conv-one-step}.
\end{proof}

\begin{lemma}[Weight perturbation]
\label{lem:alpha_perturb}
Under Assumptions~\ref{ass:smooth} and~\ref{ass:weight}, for all $\theta$ and $r$:
\begin{equation}
  \mathbb{E}\big[\|\nabla \widetilde{F}_r(\theta) - \nabla F_{*,r}(\theta)\|^2\big]
  \le \frac{G^2\,C_\alpha}{S}.
  \label{eq:conv-alpha-grad}
\end{equation}
\end{lemma}

\begin{proof}
By definition,
$\nabla \widetilde{F}_r(\theta) - \nabla F_{*,r}(\theta)
= \sum_{i \in \mathcal{B}} (\bar{\alpha}_i^r - \bar{\alpha}_{i,*}^r)\,\nabla F_i(\theta)$.
Under Assumption~\ref{ass:smooth}, the triangle inequality implies that 
\begin{equation}
  \|\nabla \widetilde{F}_r(\theta) - \nabla F_{*,r}(\theta)\|
  \le G\,\|\bar{\alpha}^r - \bar{\alpha}_*^r\|_1.
\end{equation}
Squaring both sides and taking expectations, the result follows from Assumption~\ref{ass:weight}.
\end{proof}

We now state the main convergence result.

\begin{theorem}[Convergence of \textsc{FedFG}]
\label{thm:main}
Under Assumptions~\ref{ass:smooth}--\ref{ass:weight} and the condition in Lemma~\ref{lem:descent}, for any $R \ge 1$:
\begin{equation}
\begin{aligned}
  &\frac{1}{\sum_{r=0}^{R-1}\eta_r Q}
  \sum_{r=0}^{R-1}\eta_r Q\,
  \mathbb{E}\big[\|\nabla F_{*,r}(\theta^r)\|^2\big] \\
  &\le
  \underbrace{\frac{C_0}{\sum_{r=0}^{R-1}\eta_r Q}}_{\text{initialization}}
  +
  \underbrace{\frac{C_1\,(Q\sigma^2 + Q^2\zeta^2)
    \sum_{r=0}^{R-1}\eta_r^2}
    {\sum_{r=0}^{R-1}\eta_r Q}}_{\text{noise \& heterogeneity}} \\
  &\quad +
  \underbrace{\frac{C_2\,\Delta_{0:R}
    + C_3\,R\,\delta_f}
    {\sum_{r=0}^{R-1}\eta_r Q}}_{\text{drift \& verification failure}}
  +
  \underbrace{\frac{C_4}{S}}_{\text{probe estimation}},
\end{aligned}
\label{eq:conv-rate}
\end{equation}
where $C_0 = 2(\mathbb{E}[\widetilde{F}_0(\theta^0)] - \widetilde{F}^{\inf})/c_1$ with $\widetilde{F}^{\inf} := \inf_{r,\theta}\widetilde{F}_r(\theta)$, and $C_i > 0(i=1,...,4)$ depend on $L, G, G_{\max}, B, V, C_\alpha$ but not on $R$ or $S$.
\end{theorem}

\begin{proof}
We decompose the expected function value decrease via the law of total expectation:
\begin{equation}
  \mathbb{E}\big[\widetilde{F}_r(\theta^{r+1})\big]
  = \mathbb{E}\big[\widetilde{F}_r(\theta^{r+1})\mathbf{1}_{\mathcal{E}^r}\big]
  + \mathbb{E}\big[\widetilde{F}_r(\theta^{r+1})\mathbf{1}_{\neg\mathcal{E}^r}\big].
\label{eq:conv-split}
\end{equation}

\textit{Correct verification ($\mathcal{E}^r$).}
Using Lemma~\ref{lem:descent} and the tower property,
\begin{equation}
\begin{aligned}
  &\mathbb{E}\big[\widetilde{F}_r(\theta^{r+1})\mathbf{1}_{\mathcal{E}^r}\big] \\
  &\le \mathbb{E}\big[\widetilde{F}_r(\theta^{r})\mathbf{1}_{\mathcal{E}^r}\big]
  - c_1\eta_r Q\,
    \mathbb{E}\big[\|\nabla \widetilde{F}_r(\theta^r)\|^2
    \mathbf{1}_{\mathcal{E}^r}\big] \\
  &\quad + c_2\,\eta_r^2\,(Q\sigma^2 + Q^2\zeta^2),
\end{aligned}
\label{eq:conv-good}
\end{equation}
where $\mathbb{P}(\mathcal{E}^r)\le 1$ is used to drop the indicator from the last term.

\textit{Verification failure ($\neg\mathcal{E}^r$).}
On $\neg\mathcal{E}^r$, by $L$-smoothness of $\widetilde{F}_r$, 
\begin{equation}
\widetilde{F}_r(\theta^{r+1})
\le
\widetilde{F}_r(\theta^r)+\langle \nabla \widetilde{F}_r(\theta^r),\Delta_{\mathcal{M}}^r\rangle+\frac{L}{2}\|\Delta_{\mathcal{M}}^r\|^2.
\end{equation}
Then by Cauchy--Schwarz and Assumption~\ref{ass:smooth},
\begin{equation}
\langle \nabla \widetilde{F}_r(\theta^r),\Delta_{\mathcal{M}}^r\rangle
\le
\|\nabla \widetilde{F}_r(\theta^r)\|\,\|\Delta_{\mathcal{M}}^r\|
\le
G\,\|\Delta_{\mathcal{M}}^r\|.
\end{equation}
Under Assumption~\ref{ass:verif}, we have
\begin{equation}
\begin{aligned}
  &\mathbb{E}\big[\widetilde{F}_r(\theta^{r+1})\mathbf{1}_{\neg\mathcal{E}^r}\big] \\
  &\le \mathbb{E}\big[\widetilde{F}_r(\theta^{r})\mathbf{1}_{\neg\mathcal{E}^r}\big] +\Big(GG_{\max}+\tfrac{L}{2}G_{\max}^2\Big)\mathbb{P}(\neg\mathcal{E}^t)\\
  &\le \mathbb{E}\big[\widetilde{F}_r(\theta^{r})\mathbf{1}_{\neg\mathcal{E}^r}\big]
  + \underbrace{\big(G\,G_{\max}
    + \tfrac{L}{2}G_{\max}^2\big)}_{=:\,C_f}\,\delta_f.
\end{aligned}
\label{eq:conv-bad}
\end{equation}

\textit{Per-round recursion.}
Adding~\eqref{eq:conv-good} and~\eqref{eq:conv-bad}, and bridging consecutive objectives via the drift bound in Assumption~\ref{ass:weight} (i.e., $\mathbb{E}[\widetilde{F}_{r+1}(\theta^{r+1})] \le \mathbb{E}[\widetilde{F}_r(\theta^{r+1})] + \Delta_r$), we obtain
\begin{equation}
\begin{aligned}
  &\mathbb{E}\big[\widetilde{F}_{r+1}(\theta^{r+1})\big] \\
  &\le \mathbb{E}\big[\widetilde{F}_r(\theta^{r})\big]
  - c_1\eta_r Q\,
    \mathbb{E}\big[\|\nabla \widetilde{F}_r(\theta^r)\|^2
    \mathbf{1}_{\mathcal{E}^r}\big] \\
  &\quad + c_2\eta_r^2(Q\sigma^2 + Q^2\zeta^2)
  + C_f\,\delta_f + \Delta_r.
\end{aligned}
\label{eq:conv-recursion}
\end{equation}

\textit{Telescoping.}
Summing~\eqref{eq:conv-recursion} over $r=0,\dots,R{-}1$ and using $\widetilde{F}_R(\theta^R)\ge \widetilde{F}^{\inf}$:
\begin{equation}
\begin{aligned}
  &c_1\sum_{r=0}^{R-1}\eta_r Q\,
    \mathbb{E}\big[\|\nabla \widetilde{F}_r(\theta^r)\|^2
    \mathbf{1}_{\mathcal{E}^r}\big] \\
  &\le \mathbb{E}[\widetilde{F}_0(\theta^0)] - \widetilde{F}^{\inf}
  + c_2(Q\sigma^2{+}Q^2\zeta^2)\!\sum_{r=0}^{R-1}\!\eta_r^2 \\
  &\quad + R\,C_f\,\delta_f + \Delta_{0:R}.
\end{aligned}
\label{eq:conv-tele}
\end{equation}

\textit{Removing the indicator.}
By Assumption~\ref{ass:smooth}, $\|\nabla \widetilde{F}_r(\theta^r)\|^2\mathbf{1}_{\neg\mathcal{E}^r}\le G^2\mathbf{1}_{\neg\mathcal{E}^r}$. Hence,
\begin{equation}
  \mathbb{E}\big[\|\nabla \widetilde{F}_r(\theta^r)\|^2
    \mathbf{1}_{\mathcal{E}^r}\big]
  \ge \mathbb{E}\big[\|\nabla \widetilde{F}_r(\theta^r)\|^2\big]
  - G^2\delta_f.
\label{eq:conv-remove-ind}
\end{equation}
Substituting~\eqref{eq:conv-remove-ind} into~\eqref{eq:conv-tele} and letting $C_f' := C_f + c_1 G^2$ absorb the additional term yields
\begin{equation}
\begin{aligned}
  &\sum_{r=0}^{R-1}\eta_r Q\,
    \mathbb{E}\big[\|\nabla \widetilde{F}_r(\theta^r)\|^2\big] \\
  &\le \frac{1}{c_1}\big(\mathbb{E}[\widetilde{F}_0(\theta^0)]
    {-} \widetilde{F}^{\inf}\big)
  + \frac{c_2}{c_1}(Q\sigma^2{+}Q^2\zeta^2)\!\sum_{r=0}^{R-1}\eta_r^2 \\
  &\quad + \frac{C_f'\,R\,\delta_f}{c_1}
  + \frac{\Delta_{0:R}}{c_1}.
\end{aligned}
\label{eq:conv-emp-bound}
\end{equation}

\textit{Connecting to the population objective.}
By the inequality $\|a\|^2 \le 2\|b\|^2 + 2\|a{-}b\|^2$ and Lemma~\ref{lem:alpha_perturb},
\begin{equation}
  \mathbb{E}\big[\|\nabla F_{*,r}(\theta^r)\|^2\big]
  \le 2\,\mathbb{E}\big[\|\nabla \widetilde{F}_r(\theta^r)\|^2\big]
  + \frac{2G^2 C_\alpha}{S}.
\label{eq:conv-pop-vs-emp}
\end{equation}
Multiplying~\eqref{eq:conv-pop-vs-emp} by $\eta_r Q$, summing over $r=0,\dots,R{-}1$, dividing by $\sum_{r=0}^{R-1}\eta_r Q$, and substituting~\eqref{eq:conv-emp-bound} completes the proof with $C_1 = 2c_2/c_1$, $C_2 = 2/c_1$, $C_3 = 2C_f'/c_1$, and $C_4 = 2G^2 C_\alpha$.
\end{proof}

\begin{remark}[Convergence rate interpretation]
\label{rem:rate}
The bound in~\eqref{eq:conv-rate} comprises five interpretable terms: (i) the initial optimality gap, decaying with the total effective step budget; (ii) irreducible noise from stochastic gradients and client heterogeneity; (iii) objective drift caused by evolving aggregation weights; (iv) a penalty for verification failure; and (v) weight estimation error from finite synthetic probes. With a constant step size $\eta_r = \eta = \Theta(1/\sqrt{QR})$, the first two terms yield the standard $\mathcal{O}(1/\sqrt{QR})$ rate for nonconvex federated optimization. The remaining terms vanish as the weights stabilize ($\Delta_{0:R}/\sqrt{QR} \to 0$), the verification mechanism succeeds with increasing probability ($\delta_f \to 0$), and the probe budget grows ($S \to \infty$), thereby recovering the canonical rate.
\end{remark}

\section{EXPERIMENTAL ANALYSIS}\label{sec4}

\begin{table*}[t]
  \centering
  \caption{The overall comparison of the compared methods
  (\cmark\ Supported \xmark\ Not Supported)}
  \label{tab:overall-method}
  \renewcommand{\arraystretch}{1.8}
  \setlength{\tabcolsep}{4pt}
  \begin{tabular}{c|c|c|c|c|c|c|c|c|c|c}
    \hline
    \textit{Method} & FedAVG & DP & FedCG & Median & TrimmedMean & Geometric & FoolsGold & GAN-Filter & DPR-PPFL & FedFG(Ours) \\
    \hline
    Privacy Preservable       & \xmark & \cmark & \cmark & \xmark & \xmark & \xmark & \xmark & \xmark & \cmark & \cmark \\
    Model Poisoning Resistible & \xmark & \xmark & \xmark & \cmark & \cmark & \cmark & \cmark & \cmark & \cmark & \cmark \\
    \hline
  \end{tabular}
\end{table*}

In this section, we evaluate the performance of FedFG in terms of privacy protection and robustness to data poisoning on three representative federated learning datasets.

\subsection{Experiment Settings}
\subsubsection{Datasets and Model}
We evaluate FedFG on three standard benchmark datasets: MNIST, FMNIST (Fashion-MNIST), and CIFAR10.

\begin{itemize}
\item \textbf{MNIST} \cite{lecun2002gradient} is a classic 10-class handwritten digit dataset (0--9) with 60{,}000 training and 10{,}000 test grayscale images of size \(28 \times 28\). In our experiments, images are resized to \(32 \times 32\) to unify the input resolution.
\item \textbf{FMNIST} \cite{xiao2017fashion} is a balanced 10-class fashion image dataset with 60{,}000 training and 10{,}000 test grayscale images of size \(28 \times 28\). We also resize images to \(32 \times 32\) to ensure a consistent model input.
\item \textbf{CIFAR10} \cite{krizhevsky2009learning} is a widely used 10-class object recognition benchmark containing 50{,}000 training and 10{,}000 test RGB images of size \(32 \times 32\), with a balanced class distribution.
\end{itemize}

To simulate realistic non-IID scenarios in FL, we partition the training data across clients using a Dirichlet distribution following prior work \cite{hsu2019measuring}. We set $\beta = 0.5$ to distribute data unevenly across clients, resulting in MNIST-0.5, FMNIST-0.5, and CIFAR10-0.5. We further construct MNIST-0.2 (i.e., $\beta = 0.2$) to model more extreme data heterogeneity among clients.

In our experiments, we adopt the canonical LeNet-5 CNN \cite{lecun2002gradient} as the backbone model. We treat the first two convolutional layers as the private extractor and the last three fully connected layers as the public classifier. Our flow-matching model architecture is adapted from TorchCFM \cite{tong2024improving,tong2023simulation}; specifically, we adjust the vector field dimensionality and step size to match the output dimensions of the extractor and the generator.

\subsubsection{Baseline Methods}

To enable a comprehensive evaluation of both privacy protection and resistance to data poisoning, we compare our approach with nine representative state-of-the-art baselines summarized in Table~\ref{tab:overall-method}. Among them, FedAVG \cite{mcmahan2017communication} serves as the canonical FL benchmark without any explicit defense. The remaining methods are grouped into three categories:
\begin{itemize}
    \item \textbf{Privacy-protective defenses}, including Differential Privacy (DP) \cite{wei2020federated}, which mitigates data leakage by injecting Gaussian noise into local models (two noise levels: $\sigma^{2}=0.1$ and $\sigma^{2}=0.001$), and FedCG \cite{ijcai2022-324}, which splits CNNs into a private feature extractor and a classifier and uses cGAN-based imitation for the latter part of the local model.
    \item \textbf{Poisoning-resistant defenses}, including coordinate-wise Median and TrimmedMean \cite{yin2018byzantine}, where the server takes the median or removes a fixed fraction of extreme values before averaging; Geometric aggregation \cite{pillutla2022robust}, which reduces sensitivity to outliers via geometric-type combination of updates; FoolsGold \cite{fang2020local}, which down-weights clients exhibiting high gradient similarity across rounds under the assumption of colluding attackers; and GAN-Filter \cite{zafar2025robust}, which employs a server-side cGAN to synthesize validation data for detecting and filtering poisoned updates without relying on external datasets.
    \item \textbf{Defenses with both capabilities}, represented by DPR-PPFL \cite{chen2024data}, which leverages representational similarity analysis to support asymmetric local models against data inversion while enabling the server to identify benign updates and robustly aggregate them under poisoning.
\end{itemize}

\subsubsection{Experimental Environment}
All experiments were conducted in Python and PyTorch on a server equipped with four NVIDIA GeForce RTX 4090 GPUs and Intel(R) Xeon(R) Platinum 8352V CPU@2.10GHz.

\begin{figure*}[!htbp] 
  \centering
  \includegraphics[width=0.9\textwidth]{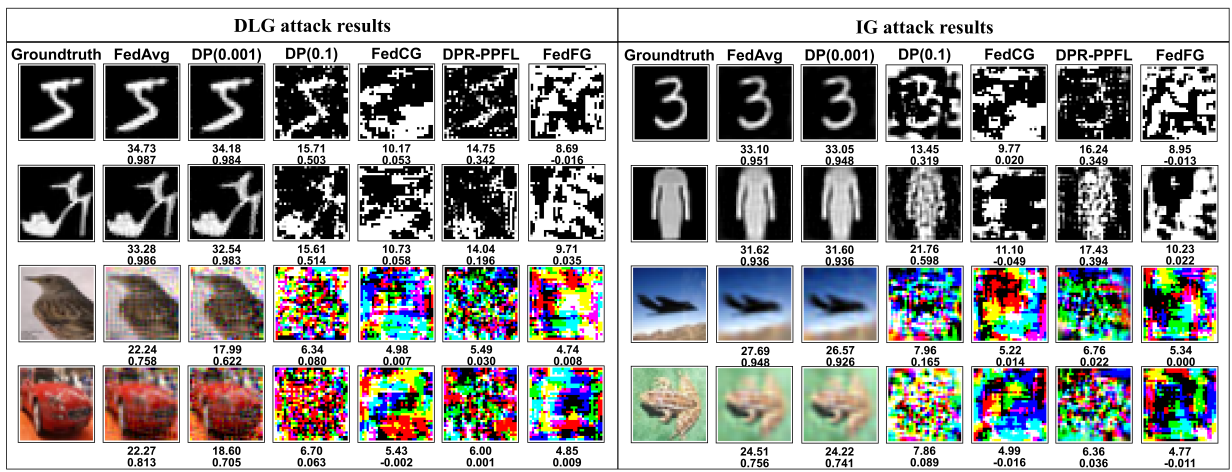}
  \caption{Comparison of ground truth images and recovered images under DLG and IG attacks. The PSNR (in dB) and SSIM values are listed sequentially below each image. The first row reports results on MNIST, the second row reports results on FMNIST, and the third and fourth rows report results on CIFAR10.}
  \label{Fig-privacy}
\end{figure*}

\subsection{Evaluation on Privacy Protection Ability}
FedFG provides privacy protection by updating the flow-matching generator $FG_i^{t}$ rather than directly updating the private extractor $E_i^{t}$. To evaluate its effectiveness, we conduct experiments against gradient-based privacy attacks as follows.

\subsubsection{Compared Methods}
We consider FedAVG \cite{mcmahan2017communication}, DP with different noise levels \cite{wei2020federated}, FedCG \cite{ijcai2022-324}, and DPR-PPFL \cite{chen2024data} as baselines.

\subsubsection{Attack Setup}
We employ two representative gradient inversion attacks to assess the effectiveness of our defense. Deep Leakage from Gradients (DLG) \cite{zhu2019deep} attempts to reconstruct clients' training samples from the model gradients shared with the server. Inverting Gradients (IG) \cite{geiping2020inverting} improves upon DLG by matching the observed and reconstructed gradients using a cosine-similarity-based objective, typically yielding higher-quality reconstructions. To strengthen the attacks in our evaluation, we add a total variation (TV) \cite{rudin1992nonlinear} regularization term to the loss function for both methods and apply them with customized hyperparameter settings. Specifically, all attacks use the Adam optimizer; DLG runs for 10{,}000 iterations and IG runs for 5{,}000 iterations.

\subsubsection{Experimental Results}
Fig.~\ref{Fig-privacy} reports both qualitative reconstructions and quantitative metrics (PSNR/SSIM; lower values indicate stronger privacy protection). The left part corresponds to DLG, and the right part corresponds to IG.
Without any defense, FedAvg exhibits the most severe leakage: reconstructed samples remain visually close to the ground truth, with consistently high PSNR and SSIM. The protection offered by DP strongly depends on the noise level. With a small noise scale ($\sigma^{2}{=}0.001$), reconstructions are still highly recognizable and the metric reduction is marginal, indicating that the perturbation is insufficient to suppress inversion. Increasing the noise to $\sigma^{2}{=}0.1$ substantially degrades reconstructed images and leads to a clear drop in PSNR and SSIM. Nevertheless, this protection is not uniform on grayscale datasets: for MNIST and FMNIST, DP ($\sigma^{2}{=}0.1$) mostly remains in the range of 15--21\,dB with SSIM above 0.5, suggesting that structural cues may still be exploitable, and faint contour details can be discerned in the recovered images. DPR-PPFL further disrupts the inversion process, leading to visibly corrupted, blob-like reconstructions with reduced metrics, although some cases still retain non-negligible PSNR and SSIM. In contrast, FedCG and FedFG consistently produce near-noise reconstructions with lower PSNR and SSIM.

Although neither FedCG nor FedFG transmits the extractor, the generator can still act as a prior over the feature distribution for gradient inversion attacks. In FedCG, the cGAN tends to learn a more concentrated conditional feature manifold, which strengthens the prior constraints and yields more discriminative structure for inversion. By contrast, FedFG learns a time-dependent vector field via flow matching and performs ODE-based sampling; the resulting generated distribution is typically smoother, making gradient matching more underconstrained and the reconstruction process less stable. Consequently, under both DLG and IG, FedFG achieves lower reconstruction-quality scores, indicating stronger privacy protection.

 \begin{figure}[htbp]
  \centering
  \subfigure[SF]{%
    \includegraphics[width=0.45\linewidth]{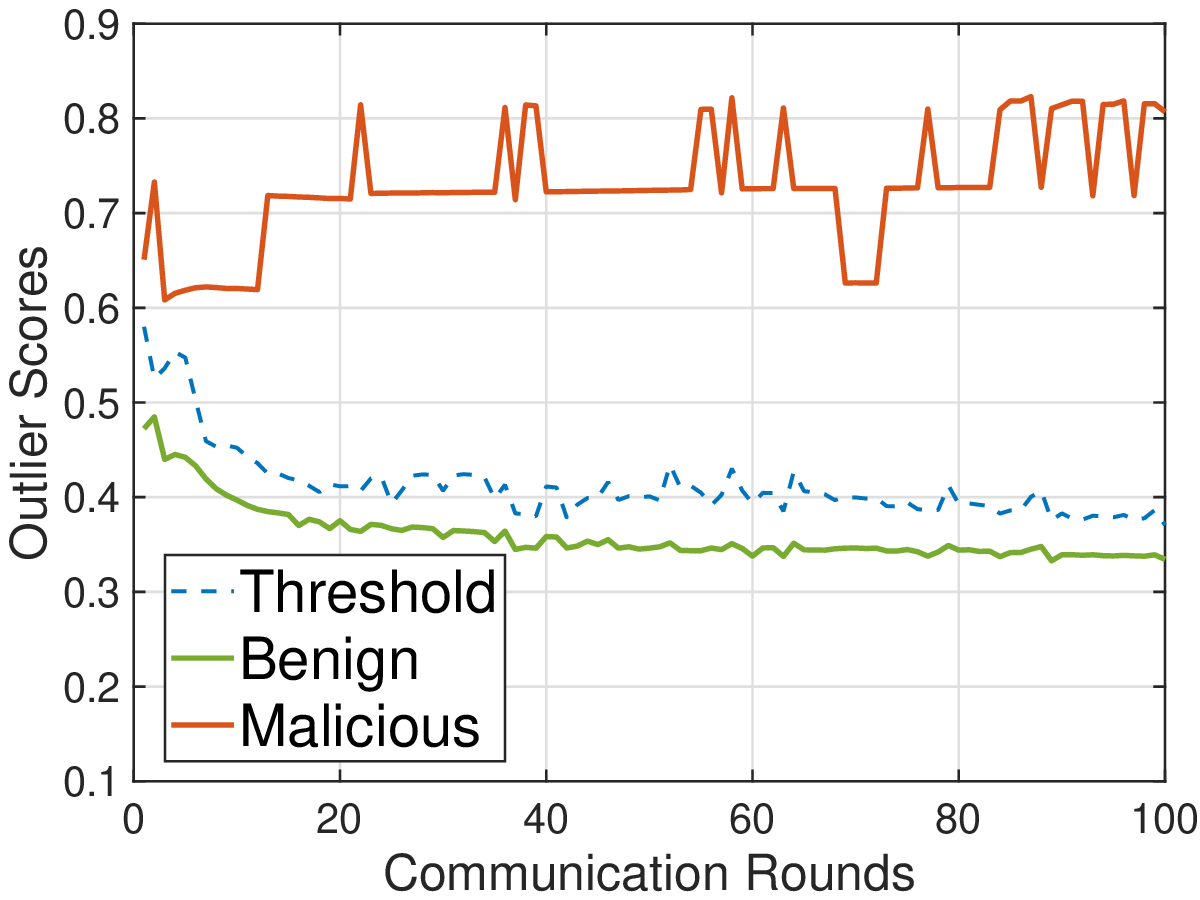}}
  \subfigure[IPM]{%
    \includegraphics[width=0.45\linewidth]{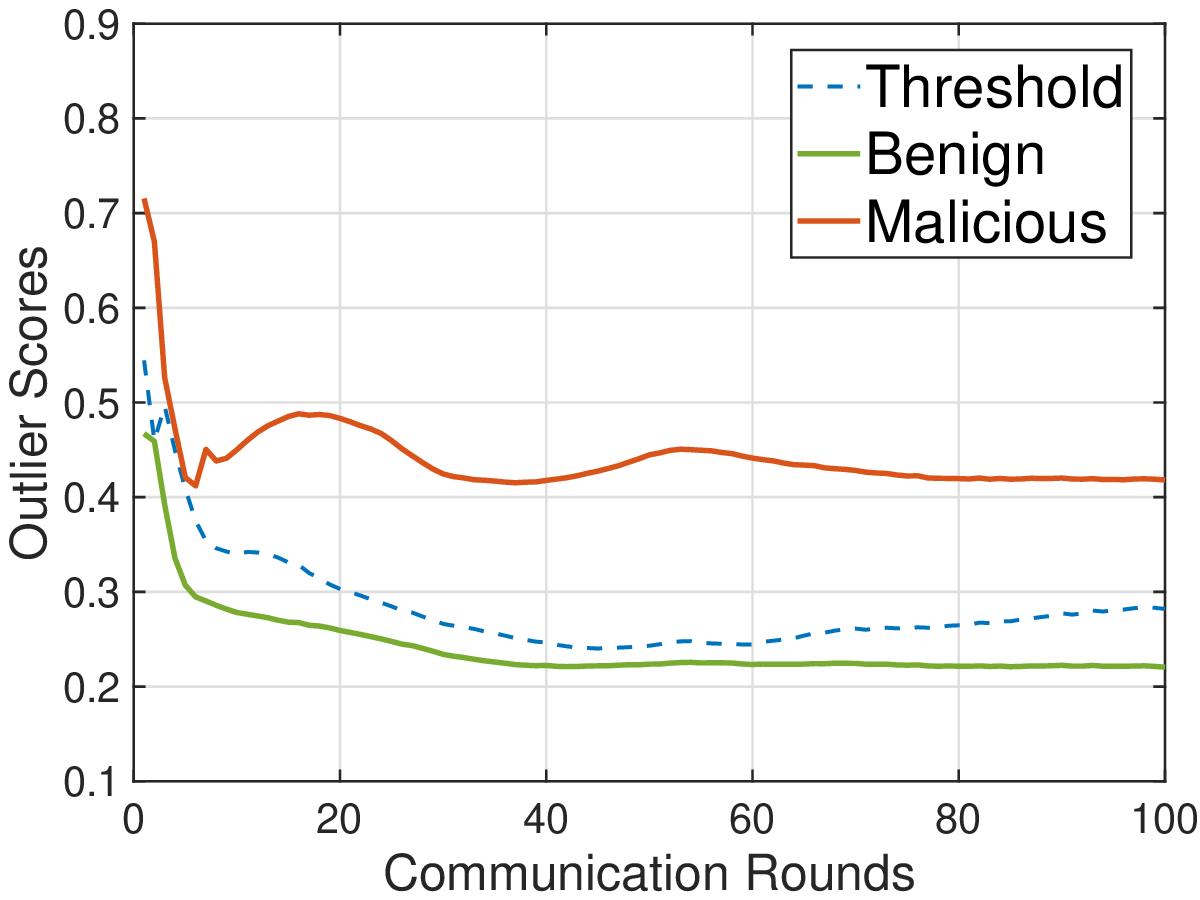}}
    \caption{The statistical results of the outlier scores for malicious and benign clients and dynamic thresholds under attacks vary with epochs on MNIST.}
  \label{Fig-score}
\end{figure}

\subsection{Evaluation on Model Poisoning Resistibility}

\subsubsection{Compared Methods}
We compare FedFG with six baselines that are designed to resist model poisoning: Median \cite{yin2018byzantine}, TrimmedMean \cite{yin2018byzantine}, Geometric \cite{pillutla2022robust}, FoolsGold \cite{fang2020local}, DPR-PPFL \cite{chen2024data}, and GAN-Filter \cite{zafar2025robust}.

\subsubsection{Attack Setup}
To assess robustness, we implement three representative poisoning attacks that cover both gradient-level and model-level adversarial behaviors, as summarized below:
\begin{itemize}
    \item \textbf{Sign Flipping (SF)} \cite{karimireddy2021learning}: Malicious clients reverse the sign of their local updates to push the aggregated global update in the opposite direction.
    \item \textbf{Inner Product Manipulation (IPM)} \cite{xie2020fall}: Malicious clients manipulate the direction of the aggregated update by making its inner product with the true gradient small or negative.
    \item \textbf{Model Poisoning Attack based on Fake clients (MPAF)} \cite{cao2022mpaf}: Fake clients steer global updates toward an attacker-chosen low-accuracy base model and scale their updates to maximally degrade overall performance.
\end{itemize}

In our implementation, the MAD-based cutoff multiplier $\gamma$ is set to $3$, corresponding to the classical three-scale cutoff, and the accuracy score threshold $\kappa$ is empirically set to $\frac{1}{2N}$, where $N$ denotes the number of clients. Following common FL evaluation practice, we instantiate an FL system consisting of 10 clients and one central server. Each client performs local training for 10 epochs using a batch size of 64 and a learning rate of 0.001. We further vary the fraction of malicious clients in $\{10\%,20\%,30\%\}$. All experiments run for 100 global communication rounds; unless otherwise specified, adversaries start behaving maliciously from round 20.

\begin{figure*}[htbp]
 \centering
 \subfigure[MNIST, SF]{\includegraphics[width=5cm]{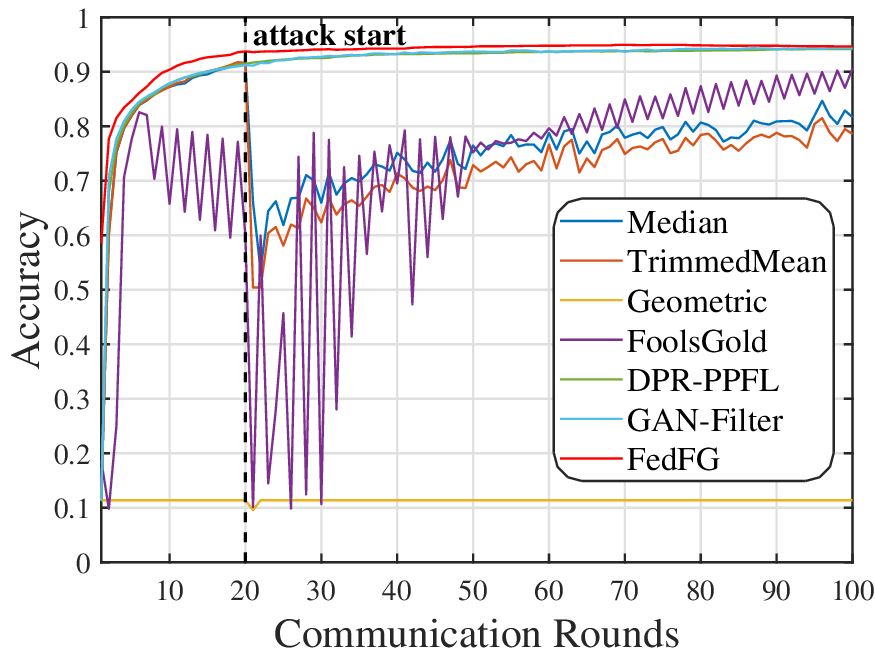}}
 \subfigure[MNIST, IPM]{\includegraphics[width=5cm]{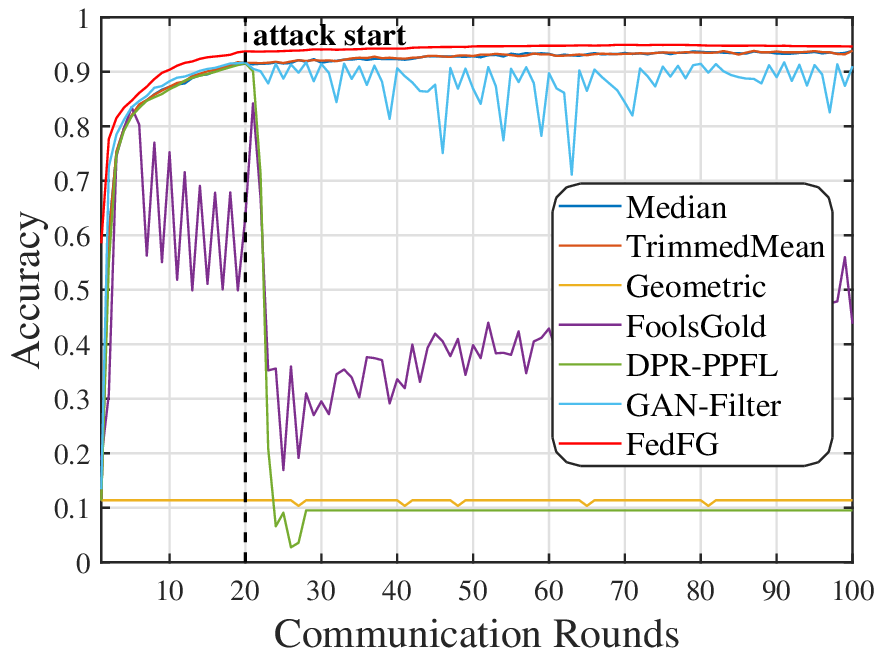}}
 \subfigure[MNIST, MPAF]{\includegraphics[width=5cm]{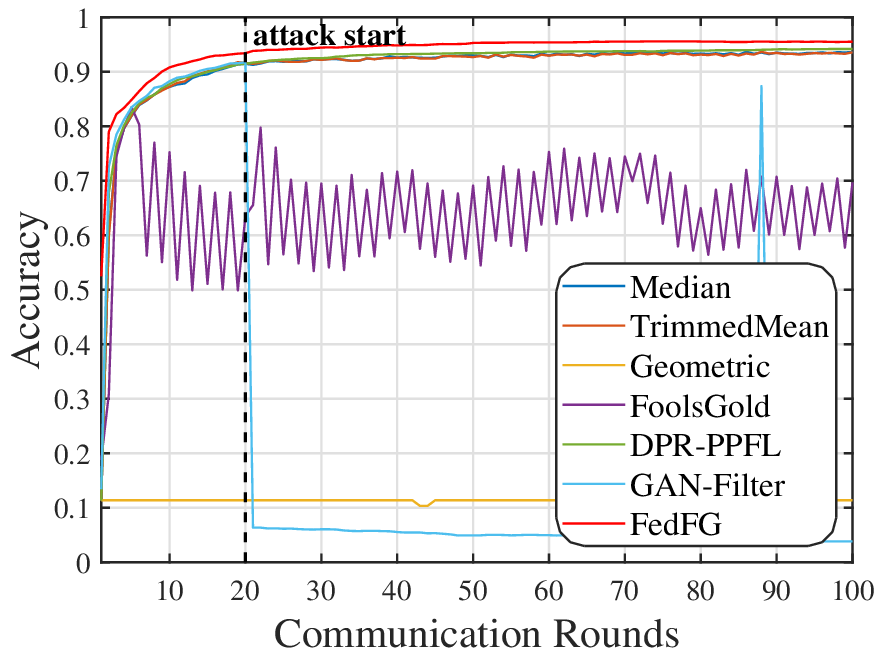}}
 \subfigure[MNIST-0.5, SF]{\includegraphics[width=5cm]{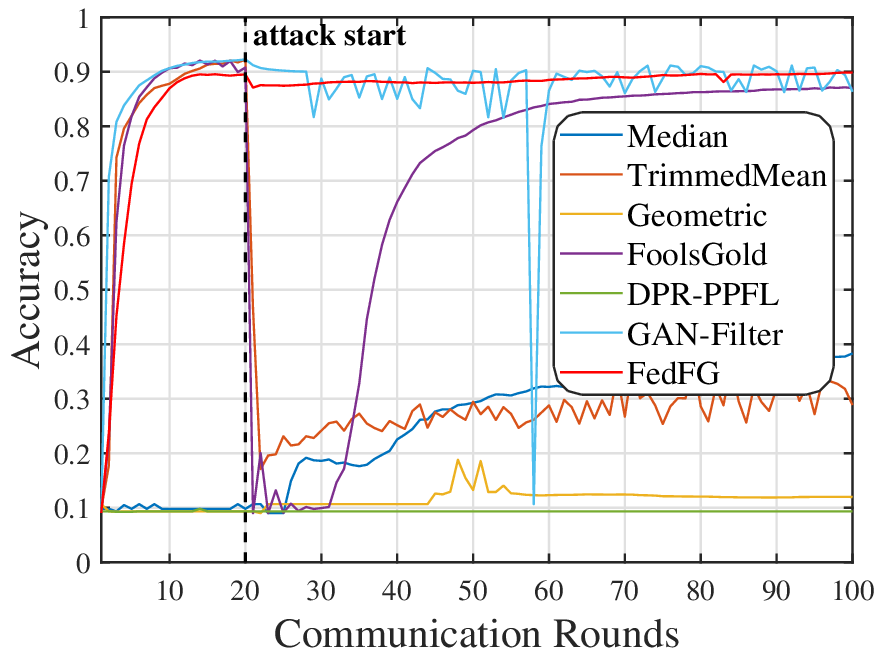}}
 \subfigure[MNIST-0.5, IPM]{\includegraphics[width=5cm]{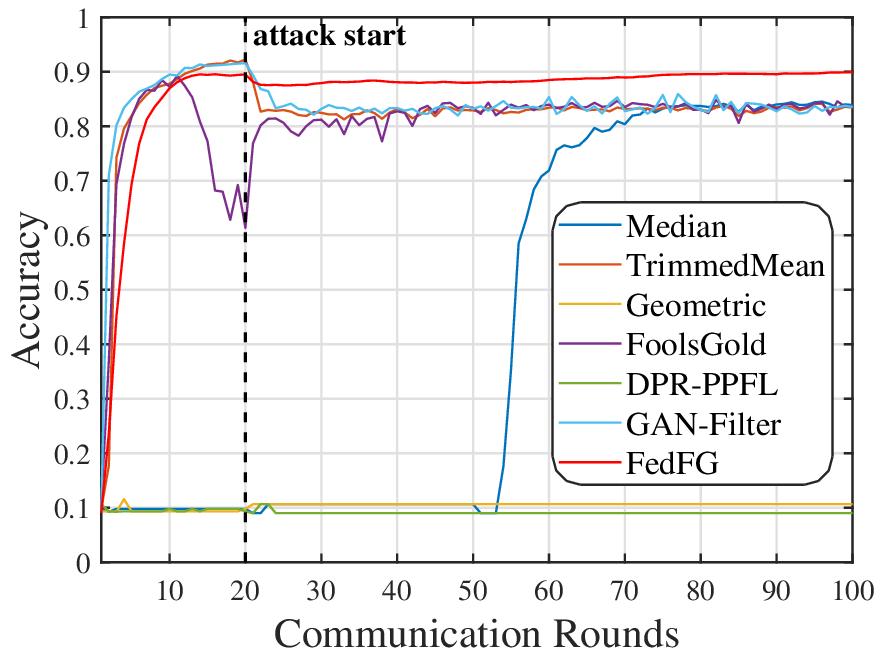}}
 \subfigure[MNIST-0.5, MPAF]{\includegraphics[width=5cm]{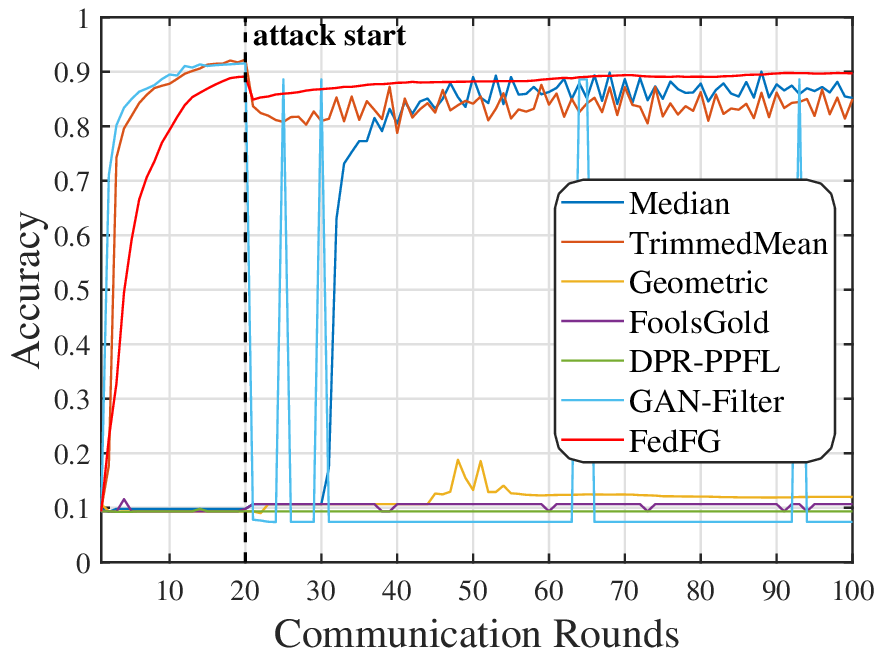}}
 \subfigure[MNIST-0.2, SF]{\includegraphics[width=5cm]{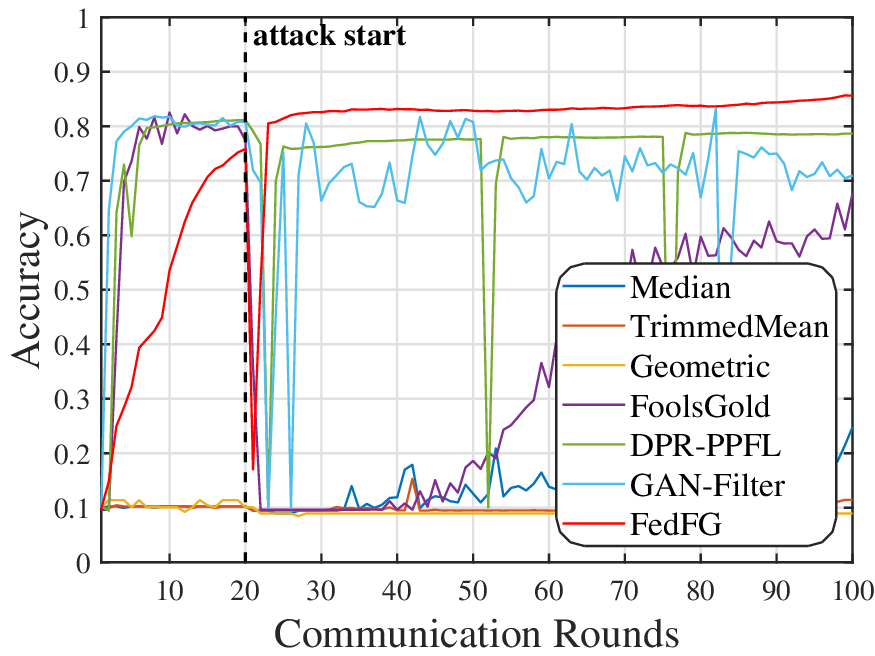}}
 \subfigure[MNIST-0.2, IPM]{\includegraphics[width=5cm]{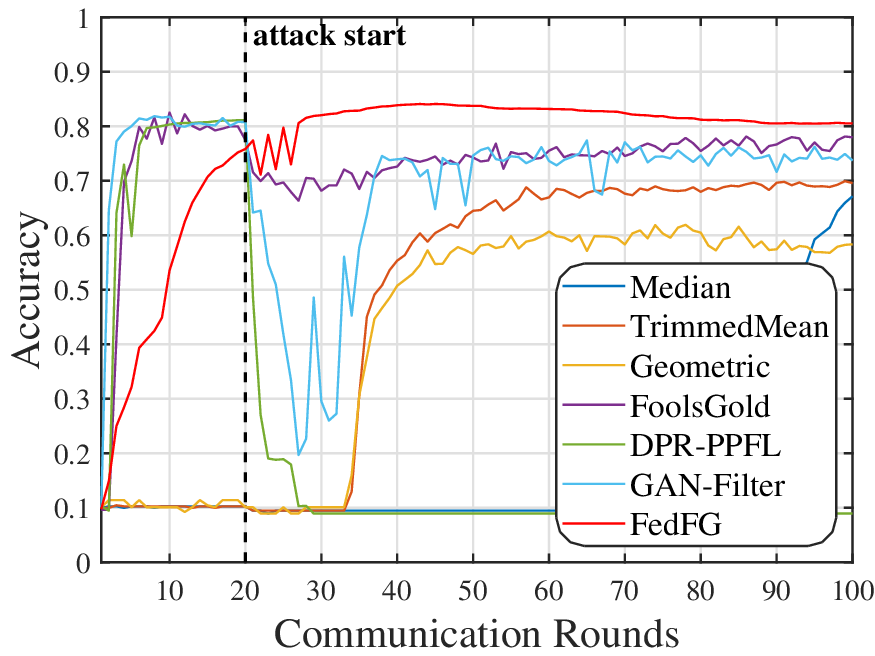}}
 \subfigure[MNIST-0.2, MPAF]{\includegraphics[width=5cm]{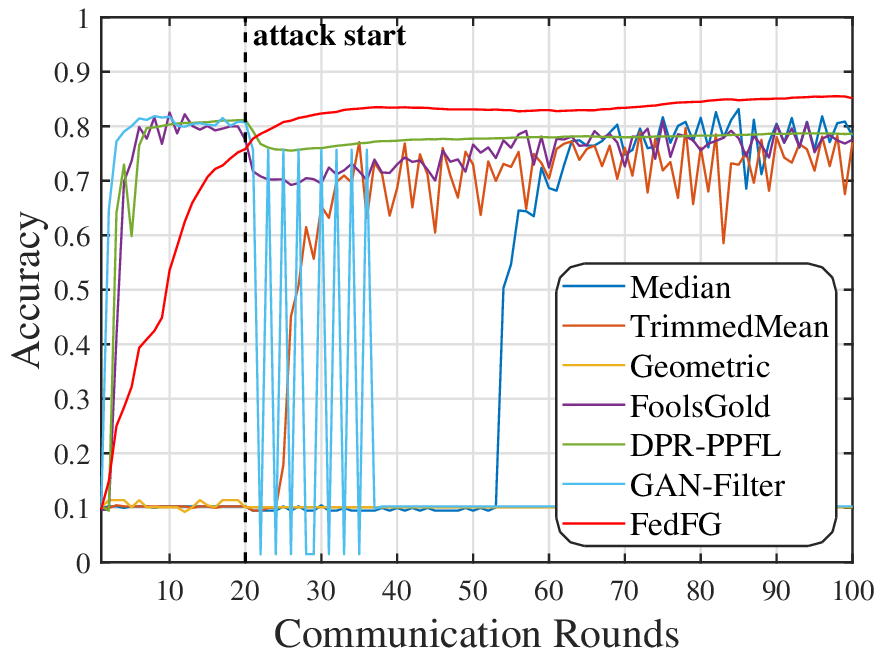}}
 \caption{Test-accuracy trajectories of FL defense methods under different distribution settings, with the attack starting at round 20.}
 \label{Fig-allpoints}
 \vspace{\baselineskip}
\end{figure*}

\subsubsection{Preliminary Verification}

Fig.~\ref{Fig-score} provides a preliminary verification of FedFG by analyzing the evolution of clients' outlier scores over 100 global communication rounds on MNIST under SF and IPM, where 30\% of participating clients are malicious. In this stress test, adversaries launch attacks from communication round 0 to evaluate the detection mechanism throughout the entire training trajectory.

As shown in Fig.~\ref{Fig-score}(a), under SF, malicious clients consistently exhibit substantially larger and more volatile outlier scores than benign clients across rounds. The MAD-based threshold adapts to the score statistics while remaining above the benign maximum and below the malicious minimum, thereby maintaining a clear separation margin. Under IPM in Fig.~\ref{Fig-score}(b), the gap between malicious and benign scores is comparatively narrower, reflecting the more subtle manipulation pattern of IPM; nevertheless, the threshold remains between the maximum benign score and the minimum malicious score over the entire training process. This persistent ordering indicates that the proposed score-and-threshold mechanism can reliably distinguish malicious updates from benign ones even when attacks start from round 0.

Overall, Fig.~\ref{Fig-score} shows that the MAD-driven dynamic threshold effectively tracks the evolving score distribution and continuously separates malicious clients from benign clients in both SF and IPM scenarios. This enables FedFG to detect and filter compromised clients on the fly, providing an initial empirical validation of its robustness and supporting its use as a robust aggregation component.

\begin{table*}[htbp]
\centering
\small
\setlength{\tabcolsep}{9pt}
\renewcommand{\arraystretch}{0.8}
\caption{Performance (ACC) under different attack rates $\epsilon$ on MNIST, FMNIST and CIFAR10.}
\begin{tabular}{llccccccccc}
\specialrule{0.08em}{0pt}{0pt}
\toprule
\toprule
& &
\multicolumn{3}{c}{$\epsilon = 10\%$} &
\multicolumn{3}{c}{$\epsilon = 20\%$} &
\multicolumn{3}{c}{$\epsilon = 30\%$} \\
\cmidrule(lr){3-5} \cmidrule(lr){6-8} \cmidrule(lr){9-11}
&
\multicolumn{1}{c}{\textbf{Methods}} &
\multicolumn{1}{c}{\textbf{SF}} &
\multicolumn{1}{c}{\textbf{IPM}} &
\multicolumn{1}{c}{\textbf{MPAF}} &
\multicolumn{1}{c}{\textbf{SF}} &
\multicolumn{1}{c}{\textbf{IPM}} &
\multicolumn{1}{c}{\textbf{MPAF}} &
\multicolumn{1}{c}{\textbf{SF}} &
\multicolumn{1}{c}{\textbf{IPM}} &
\multicolumn{1}{c}{\textbf{MPAF}} \\
\midrule

\multirow{8}{*}{\rotatebox[origin=c]{90}{MNIST}}
& FedAvg                  & 0.1742 & 0.8949 & 0.0964 & 0.0950 & 0.1050 & 0.0960 & 0.0831 & 0.0957 & 0.0954  \\
& Median                  & 0.9362 & 0.9422 & 0.9418 & 0.9013 & 0.9390 & 0.9408 & 0.8166 & 0.9383 & 0.9363  \\
& TrimmedMean             & 0.9118 & 0.9413 & 0.9418 & 0.8655 & 0.9390 & 0.9400 & 0.7857 & 0.9386 & 0.9357  \\
& Geometric               & 0.1104 & 0.1104 & 0.1104 & 0.1113 & 0.1113 & 0.1113 & 0.1137 & 0.1137 & 0.1137  \\
& FoolsGold               & 0.7073 & 0.5596 & 0.4802 & 0.8863 & 0.5230 & 0.5115 & 0.9060 & 0.4377 & 0.6969  \\
& DPR-PPFL                & 0.9416 & 0.9378 & 0.9416 & 0.9420 & 0.9085 & 0.9420 & 0.9417 & 0.0951 & 0.9417  \\
& GAN-Filter              & 0.9451 & 0.9400 & 0.0964 & 0.9435 & 0.9293 & 0.0960 & 0.9434 & 0.9100 & 0.0383  \\
\cmidrule(lr){2-11}
& FedFG (Ours)            & \textbf{0.9493} & \textbf{0.9496} & \textbf{0.9558} & \textbf{0.9490} & \textbf{0.9490} & \textbf{0.9558} & \textbf{0.9466} & \textbf{0.9463} & \textbf{0.9551} \\
\midrule
\multirow{8}{*}{\rotatebox[origin=c]{90}{FMNIST}}
& FedAvg                  & 0.3631 & 0.8122 & 0.0953 & 0.0988 & 0.1035 & 0.0963 & 0.0983 & 0.0983 & 0.0966  \\
& Median                  & 0.8113 & 0.8227 & 0.8249 & 0.7385 & 0.8238 & 0.8165 & 0.3106 & 0.8217 & 0.8146  \\
& TrimmedMean             & 0.7973 & 0.8273 & 0.8273 & 0.7320 & 0.8248 & 0.8200 & 0.6746 & 0.8180 & 0.8163  \\
& Geometric               & 0.7967 & 0.8253 & 0.8204 & 0.7055 & 0.8233 & 0.8128 & 0.5514 & 0.8226 & 0.8057  \\
& FoolsGold               & 0.2022 & 0.5924 & 0.4873 & 0.7733 & 0.5733 & 0.5945 & 0.7994 & 0.4834 & 0.5426  \\
& DPR-PPFL                & 0.8300 & 0.7862 & 0.8322 & 0.8273 & 0.8178 & 0.8303 & 0.8337 & 0.0983 & 0.8326  \\
& GAN-Filter              & 0.8311 & 0.8273 & 0.0953 & 0.8300 & 0.7803 & 0.0963 & 0.8286 & 0.7420 & 0.0966  \\
\cmidrule(lr){2-11}
& FedFG (Ours)            & \textbf{0.8429} & \textbf{0.8429} & \textbf{0.8449} & \textbf{0.8430} & \textbf{0.8430} & \textbf{0.8463} & \textbf{0.8417} & \textbf{0.8417} & \textbf{0.8423} \\
\midrule
\multirow{8}{*}{\rotatebox[origin=c]{90}{CIFAR10}}
& FedAvg                  & 0.0949 & 0.4473 & 0.1033 & 0.0953 & 0.1053 & 0.1053 & 0.0963 & 0.1034 & 0.1063  \\
& Median                  & 0.4893 & 0.5104 & 0.5084 & 0.4040 & 0.5240 & 0.4958 & 0.2840 & 0.5034 & 0.4866  \\
& TrimmedMean             & 0.4442 & 0.5096 & 0.5047 & 0.3993 & 0.5185 & 0.4980 & 0.2349 & 0.4989 & 0.4780  \\
& Geometric               & 0.4587 & 0.5176 & 0.5020 & 0.3073 & 0.5130 & 0.4853 & 0.1251 & 0.4883 & 0.4780  \\
& FoolsGold               & 0.1769 & 0.1929 & 0.1400 & 0.2820 & 0.2680 & 0.2148 & 0.3531 & 0.2894 & 0.2606  \\
& DPR-PPFL                & 0.5171 & 0.5136 & 0.5187 & 0.5070 & 0.5123 & 0.5218 & 0.1060 & 0.0960 & 0.5220  \\
& GAN-Filter              & 0.5147 & 0.5053 & 0.3484 & 0.5108 & 0.4785 & 0.1053 & 0.5146 & 0.3143 & 0.1063  \\
\cmidrule(lr){2-11}
& FedFG (Ours)            & \textbf{0.5611} & \textbf{0.5456} & \textbf{0.5327} & \textbf{0.5558} & \textbf{0.5478} & \textbf{0.5570} & \textbf{0.5234} & \textbf{0.5631} & \textbf{0.5571} \\
\bottomrule
\bottomrule
\specialrule{0.08em}{0pt}{0pt}
\end{tabular}
\label{tab:results1}
\end{table*}

\subsubsection{Experimental Results}

Table~\ref{tab:results1} summarizes the results on IID MNIST, FMNIST, and CIFAR10. FedFG achieves the highest accuracy across all attack types and malicious-client fractions. On MNIST, even when $\epsilon$ increases to 30\% under the strong MPAF attack, FedFG remains stable, achieving an accuracy of 0.946--0.956. Conventional robust aggregators, however, degrade to varying degrees. Median and TrimmedMean exhibit a clear accuracy drop under SF as $\epsilon$ increases. FoolsGold is sensitive to the attack type and exhibits instability. Geometric aggregation nearly collapses to random-guessing performance on MNIST (approximately 0.11) across all attacks, indicating high sensitivity under the considered threat model and training configuration. Similar trends are observed on FMNIST, where FedFG achieves an accuracy of approximately 0.842 across different attacks and values of $\epsilon$, outperforming both Median and TrimmedMean while remaining robust in settings where GAN-Filter and DPR-PPFL fail. On CIFAR10, FedFG maintains an accuracy of 0.52--0.56 across all attacks and malicious fractions, substantially exceeding the best-performing baseline.

Table~\ref{tab:results2} reports results under Dirichlet partitioning ($\beta=0.5$), which introduces client heterogeneity and further complicates robust aggregation. On MNIST-0.5, FedFG maintains high accuracy (approximately 0.90) across all attacks and values of $\epsilon$. Several baselines, however, become fragile. Median collapses to 0.1002 under IPM with $\epsilon=10\%$ and drops to 0.1051 under MPAF with $\epsilon=10\%$, while Geometric aggregation is likewise unstable. DPR-PPFL and GAN-Filter perform well under specific configurations (for example, on MNIST-0.5 with $\epsilon=20\%$, DPR-PPFL still reaches 0.9628 under MPAF, and GAN-Filter exceeds 0.91 under SF and IPM) but can fail severely under others. When $\epsilon=30\%$, for instance, DPR-PPFL collapses to near-random performance, and GAN-Filter fails under nearly all MPAF settings. These failures point to practical limitations of each method. Because DPR-PPFL relies on similarity-based filtering, its reliability can degrade under non-IID heterogeneity. GAN-Filter, on the other hand, depends on an unpoisoned global model to guide the server-side generator in producing trustworthy samples and thus can be undermined by multi-round, consistent attacks such as MPAF.

On FMNIST-0.5, FedFG achieves an accuracy of 0.760--0.781 across all attacks and values of $\epsilon$, again outperforming all other methods. Under SF with larger $\epsilon$, many defenses degrade rapidly, whereas FedFG remains substantially more stable. On CIFAR10-0.5, FedFG attains the best results in all settings. Even with $\epsilon=30\%$, FedFG remains above 0.51 under all three attacks, while other defenses are often below 0.45 or fluctuate markedly. Taken together, these results confirm that FedFG handles complex visual data, strong client heterogeneity, and adversarial poisoning more reliably than existing methods.

Fig.~\ref{Fig-allpoints} shows the test-accuracy trajectories over 100 rounds under three distribution settings: IID MNIST, non-IID MNIST-0.5, and extremely non-IID MNIST-0.2. The vertical line at round 20 marks the onset of the attack.

For IID MNIST (Fig.~\ref{Fig-allpoints}(a)--(c)), coordinate-wise robust methods are generally effective under IPM and MPAF but become unstable under SF as the adversarial contribution increases, consistent with the results in Table~\ref{tab:results1}. FoolsGold exhibits large fluctuations and is strongly affected by attacks. FedFG, in contrast, remains among the top-performing methods throughout training, suggesting that its filtering and reweighting mechanisms suppress poisoned updates while preserving useful learning signals. For MNIST-0.5 (Fig.~\ref{Fig-allpoints}(d)--(f)), client heterogeneity amplifies the damage caused by attacks. Multiple baselines drop noticeably after round 20 and recover slowly, suggesting that even a small number of poisoned updates can have a disproportionate impact under heterogeneous conditions. FedFG maintains a smooth, high-accuracy trajectory, indicating that the similarity-based filtering remains effective under Dirichlet partitioning despite the naturally increased dispersion among benign clients. For the extreme MNIST-0.2 setting (Fig.~\ref{Fig-allpoints}(g)--(i)), the task becomes even more challenging, as defenses that rely on simplistic assumptions are more likely to misclassify benign clients or fail during aggregation. Several baselines collapse to low accuracy or remain unstable for extended periods, whereas FedFG sustains high accuracy with relatively small fluctuations. These results suggest that our distribution-distance score remains discriminative even under severe heterogeneity, allowing FedFG to maintain stability across diverse client distributions and adversarial conditions.

\begin{table*}[htbp]
\centering
\small
\setlength{\tabcolsep}{9pt}
\renewcommand{\arraystretch}{0.8}
\caption{Performance (ACC) under different attack rates $\epsilon$ on MNIST-0.5, FMNIST-0.5 and CIFAR10-0.5.}
\begin{tabular}{llccccccccc}
\specialrule{0.08em}{0pt}{0pt}
\toprule
\toprule
& &
\multicolumn{3}{c}{$\epsilon = 10\%$} &
\multicolumn{3}{c}{$\epsilon = 20\%$} &
\multicolumn{3}{c}{$\epsilon = 30\%$} \\
\cmidrule(lr){3-5} \cmidrule(lr){6-8} \cmidrule(lr){9-11}
&
\multicolumn{1}{c}{\textbf{Methods}} &
\multicolumn{1}{c}{\textbf{SF}} &
\multicolumn{1}{c}{\textbf{IPM}} &
\multicolumn{1}{c}{\textbf{MPAF}} &
\multicolumn{1}{c}{\textbf{SF}} &
\multicolumn{1}{c}{\textbf{IPM}} &
\multicolumn{1}{c}{\textbf{MPAF}} &
\multicolumn{1}{c}{\textbf{SF}} &
\multicolumn{1}{c}{\textbf{IPM}} &
\multicolumn{1}{c}{\textbf{MPAF}} \\
\midrule

\multirow{8}{*}{\rotatebox[origin=c]{90}{MNIST-0.5}}
& FedAvg                  & 0.0922 & 0.3469 & 0.1022 & 0.0920 & 0.1063 & 0.1020 & 0.0900 & 0.0900 & 0.1049  \\
& Median                  & 0.9089 & 0.1002 & 0.1051 & 0.6008 & 0.8758 & 0.1063 & 0.3831 & 0.8389 & 0.8520  \\
& TrimmedMean             & 0.8447 & 0.9098 & 0.9182 & 0.5683 & 0.8258 & 0.8703 & 0.2886 & 0.8343 & 0.8497  \\
& Geometric               & 0.1051 & 0.0920 & 0.0920 & 0.4540 & 0.1063 & 0.1063 & 0.1197 & 0.1066 & 0.1066  \\
& FoolsGold               & 0.9151 & 0.2040 & 0.7380 & 0.8923 & 0.8128 & 0.8993 & 0.8700 & 0.8366 & 0.8600  \\
& DPR-PPFL                & 0.9211 & 0.3907 & \textbf{0.9271} & 0.0935 & 0.8030 & \textbf{0.9628} & 0.0934 & 0.0900 & 0.0934  \\
& GAN-Filter              & 0.9169 & \textbf{0.9242} & 0.1022 & \textbf{0.9175} & \textbf{0.9190} & 0.1020 & 0.8634 & 0.8354 & 0.0740  \\
\cmidrule(lr){2-11}
& FedFG (Ours)            & \textbf{0.9242} & 0.9240 & 0.9238 & 0.9023 & 0.9045 & 0.9065 & \textbf{0.8983} & \textbf{0.8991} & \textbf{0.8971} \\
\midrule
\multirow{8}{*}{\rotatebox[origin=c]{90}{FMNIST-0.5}}
& FedAvg                  & 0.6453 & 0.7498 & 0.1022 & 0.0998 & 0.2303 & 0.0998 & 0.1011 & 0.1011 & 0.1011  \\
& Median                  & 0.7080 & \textbf{0.7776} & 0.7464 & 0.6073 & 0.7530 & 0.7410 & 0.6529 & 0.6829 & 0.6577  \\
& TrimmedMean             & 0.6318 & 0.7609 & 0.7609 & 0.4625 & 0.7285 & 0.7300 & 0.2780 & 0.6731 & 0.7023  \\
& Geometric               & 0.7396 & 0.7540 & 0.7540 & 0.5885 & 0.7508 & 0.7398 & 0.2331 & 0.7151 & 0.7151  \\
& FoolsGold               & 0.5664 & 0.4960 & 0.6369 & 0.4960 & 0.3130 & 0.6395 & 0.7083 & 0.7214 & 0.7249  \\
& DPR-PPFL                & 0.7516 & 0.7413 & 0.7520 & 0.7450 & 0.7413 & 0.7473 & 0.7377 & 0.3549 & 0.7354  \\
& GAN-Filter              & 0.7518 & 0.7567 & 0.1022 & 0.7095 & 0.7433 & 0.0998 & 0.7274 & 0.7463 & 0.1011  \\
\cmidrule(lr){2-11}
& FedFG (Ours)            & \textbf{0.7620} & 0.7620 & \textbf{0.7780} & \textbf{0.7603} & \textbf{0.7608} & \textbf{0.7675} & \textbf{0.7731} & \textbf{0.7731} & \textbf{0.7806} \\
\midrule
\multirow{8}{*}{\rotatebox[origin=c]{90}{CIFAR10-0.5}}
& FedAvg                  & 0.1124 & 0.3542 & 0.0967 & 0.0965 & 0.2978 & 0.0963 & 0.0951 & 0.1037 & 0.0971  \\
& Median                  & 0.4318 & 0.4696 & 0.4798 & 0.3298 & 0.4680 & 0.4690 & 0.1803 & 0.4006 & 0.3991  \\
& TrimmedMean             & 0.3593 & 0.3971 & 0.4416 & 0.2173 & 0.4295 & 0.4275 & 0.1174 & 0.3683 & 0.4326  \\
& Geometric               & 0.4233 & 0.4627 & 0.4671 & 0.3405 & 0.4680 & 0.4570 & 0.1506 & 0.4146 & 0.4220  \\
& FoolsGold               & 0.2818 & 0.2136 & 0.2556 & 0.2233 & 0.2650 & 0.4535 & 0.4149 & 0.4177 & 0.4166  \\
& DPR-PPFL                & 0.4809 & 0.2571 & 0.4858 & 0.4738 & 0.2688 & 0.4913 & 0.4486 & 0.2731 & 0.4440  \\
& GAN-Filter              & 0.4442 & 0.4138 & 0.0791 & 0.4585 & 0.4913 & 0.0753 & 0.4406 & 0.4363 & 0.0971  \\
\cmidrule(lr){2-11}
& FedFG (Ours)            & \textbf{0.5080} & \textbf{0.4827} & \textbf{0.5042} & \textbf{0.5263} & \textbf{0.5143} & \textbf{0.5188} & \textbf{0.5163} & \textbf{0.5154} & \textbf{0.5223} \\
\bottomrule
\bottomrule
\specialrule{0.08em}{0pt}{0pt}
\end{tabular}
\label{tab:results2}
\end{table*}

\begin{figure*}[htbp]
 \centering
 \subfigure[FMNIST, SF]{\includegraphics[height=3cm,width=4cm]{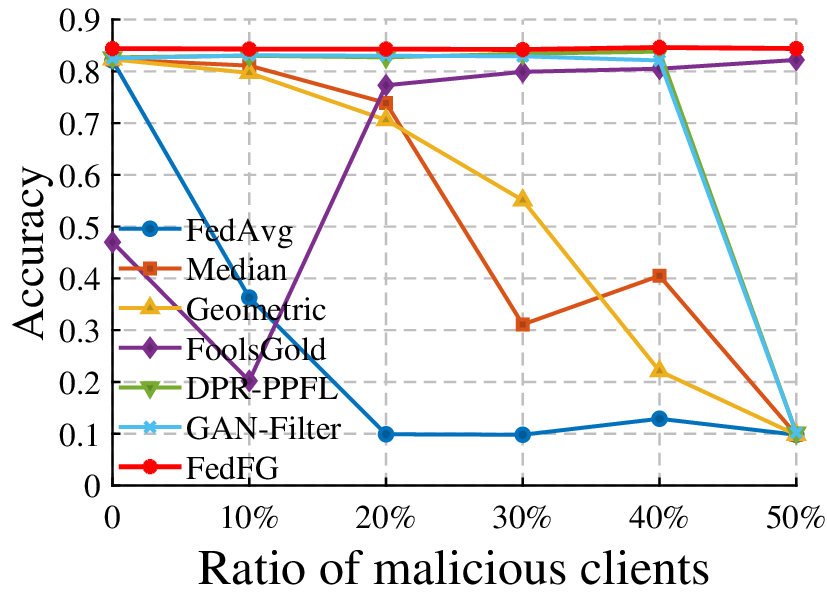}}
 \subfigure[FMNIST, IPM]{\includegraphics[height=3cm,width=4cm]{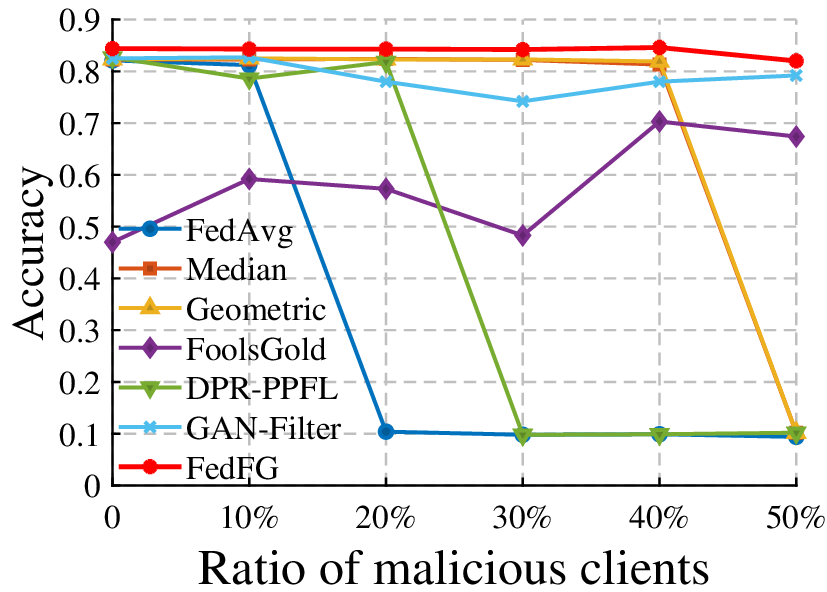}}
 \subfigure[FMNIST-0.5, SF]{\includegraphics[height=3cm,width=4cm]{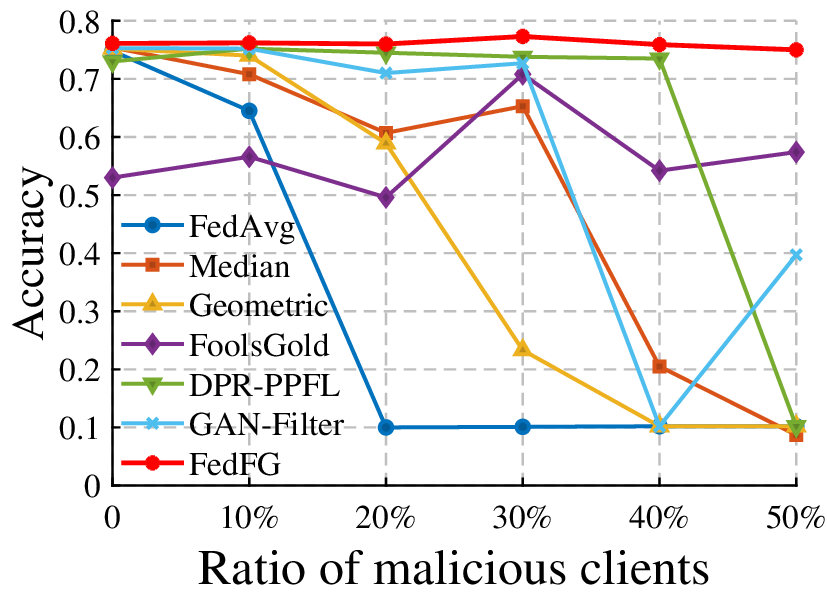}}
 \subfigure[FMNIST-0.5, IPM]{\includegraphics[height=3cm,width=4cm]{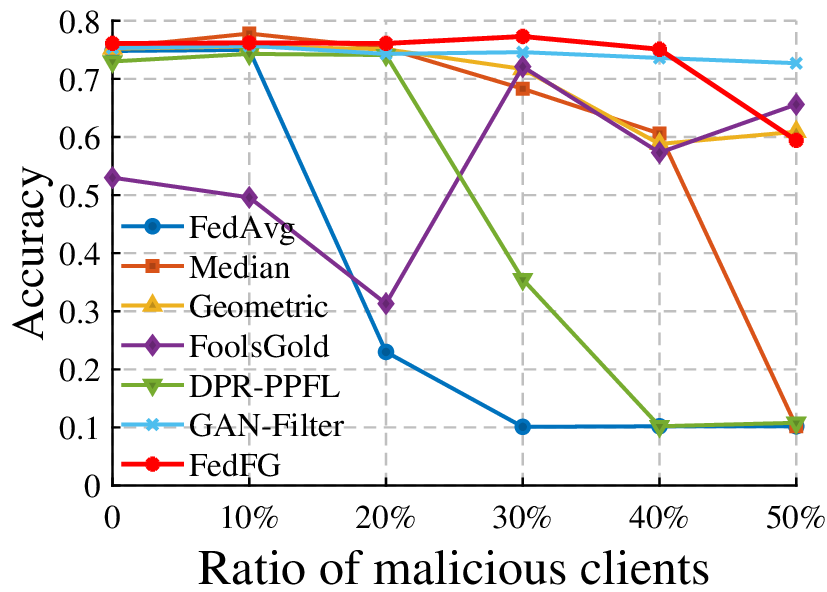}}
 \caption{Test accuracy of FL methods on FMNIST and FMNIST-0.5 under SF and IPM as the proportion of malicious clients increases.}
 \label{Fig-ratio}
 \vspace{\baselineskip}
\end{figure*}

\subsubsection{Further Robustness Analysis}

In this subsection, we vary two factors to stress-test FedFG under more demanding conditions: the proportion of malicious clients and the timing of the attack onset.

\paragraph{Impact of malicious-client fraction.}
Fig.~\ref{Fig-ratio} reports accuracy on FMNIST and FMNIST-0.5 under SF and IPM as $\epsilon$ increases from 0\% to 50\%. Most baselines degrade rapidly as $\epsilon$ grows, and some collapse to near-random performance at moderate-to-high malicious fractions. Coordinate-wise robust aggregators and similarity-based methods are particularly vulnerable, as they can fail abruptly when attackers dominate the aggregation statistics or when too few benign clients remain to form a stable consensus. FedFG, in contrast, exhibits little to no accuracy loss and maintains high accuracy across the entire range of $\epsilon$. Rather than relying on fixed trimming thresholds or brittle clustering assumptions, FedFG combines a generator-driven accuracy score with distribution-similarity detection to identify inconsistent behaviors. These two signals complement each other and together preserve robustness as $\epsilon$ increases.

\paragraph{Robustness to attacks starting at round~0.}
In the default setting, attacks begin at round~20, allowing the system to establish a benign warm-up trajectory. We further consider a stronger adversary in which poisoning begins at round~0, thereby preventing any such initialization. Fig.~\ref{Fig-bar} compares representative baselines (Median, DPR-PPFL, and GAN-Filter) with FedFG on MNIST and MNIST-0.5 under SF and IPM. When poisoning is present from the outset, several baselines suffer severe accuracy collapse, especially on heterogeneous MNIST-0.5, where early optimization dynamics are more fragile. GAN-Filter is particularly affected because it depends on an unpoisoned global model to guide the server-side generator and thus cannot cope with attacks in the initial rounds. In our experiments, FedFG fails in only one case (MNIST-0.5, SF, $\epsilon=30\%$) and maintains high accuracy in all remaining scenarios. This is consistent with the preliminary validation in Fig.~\ref{Fig-score}, where the MAD-based dynamic threshold effectively separates malicious clients from benign ones throughout training, even when attacks begin in the first communication rounds. FedFG therefore serves not only as a post-attack defense but also as a viable training strategy when poisoning is present from the start.

\begin{figure*}[!htbp]
 \centering
   \includegraphics[width=0.95\textwidth]{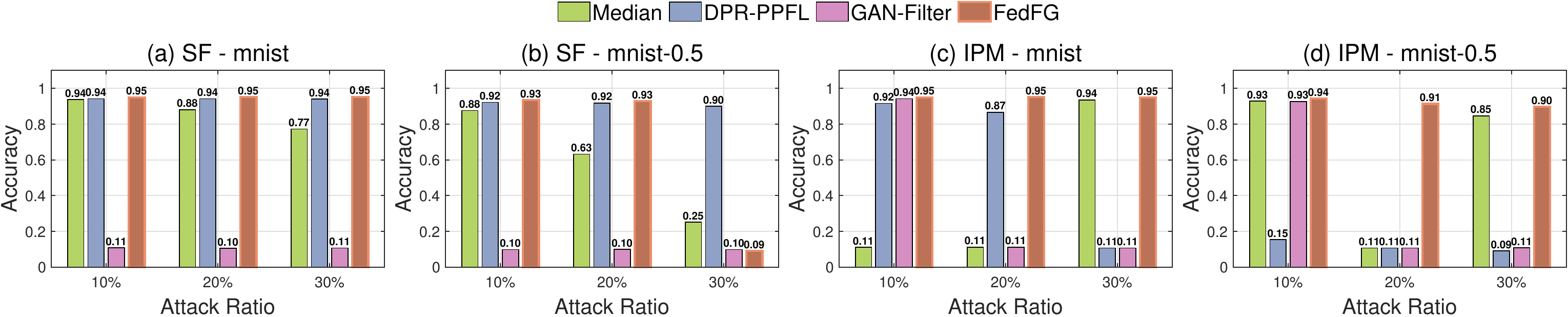}
   \caption{Accuracy comparison of Median, DPR-PPFL, GAN-Filter, and FedFG on MNIST and MNIST-0.5 under SF and IPM when the attack starts at round 0.}
   \label{Fig-bar}
\end{figure*}

Overall, the robustness of FedFG rests on two complementary signals, both computed on synthetic features produced by the global generator. The first is a probability-distribution distance that measures how each client's predictive distribution deviates from the majority; a MAD-driven adaptive threshold then identifies and filters anomalous clients. The second is an accuracy score that evaluates each client classifier on label-conditioned synthetic features and down-weights unreliable clients accordingly. This design tolerates both gradient-level and model-level attacks while preserving split-learning-style privacy, as the private feature extractor remains local. Across multiple datasets and heterogeneity settings, FedFG outperforms representative robust aggregation and privacy-preserving baselines, providing a federated optimization framework that combines privacy preservation with strong resistance to poisoning.

\section{CONCLUSION}\label{sec5}
In this paper, we propose FedFG, a robust federated learning framework based on flow-matching generation that jointly preserves client privacy and defends against sophisticated poisoning attacks. FedFG decouples each client model into a private extractor and a public classifier, and uses a flow-matching generator to replace the extractor during server communication, thereby reducing privacy leakage from shared updates. On the server side, FedFG verifies client updates using synthetic feature probes and performs robust aggregation via outlier filtering and accuracy-aware reweighting. Extensive experiments under both IID and non-IID settings show that FedFG consistently outperforms state-of-the-art defenses under multiple Byzantine attack strategies. Despite these promising results, FedFG introduces additional computation and communication overhead due to generator training and probe-based verification, which may be challenging for large-scale client populations. Future work will focus on improving efficiency for large-scale deployments and extending the framework to more adaptive threats and more complex real-world scenarios.


\bibliographystyle{IEEEtran} 
\bibliography{refs}


\end{document}